\documentclass[11pt,a4paper]{article}
\usepackage[margin=1in]{geometry}
\usepackage{tikz}
\usetikzlibrary{positioning,arrows.meta}
\usepackage{subcaption}

\usepackage[english]{babel}
\usepackage{algpseudocode}
\usepackage{makecell}
\usepackage[english]{babel}
\usepackage{array} 
\usepackage{booktabs} 
\usepackage{authblk} 

\usepackage[ruled,linesnumbered,vlined]{algorithm2e}
\usepackage{algpseudocode}
\usepackage{amssymb}
\usepackage{mathtools}

\usepackage{amsmath}
\usepackage{amsthm}
\usepackage{graphicx}
\usepackage[colorlinks=true, allcolors=blue]{hyperref}
\usepackage{natbib}
\usepackage{framed}
\definecolor{shadecolor}{gray}{0.9}
\usepackage{pgfplots}
\pgfplotsset{compat=1.17}
\usepackage{tikz}
 \usetikzlibrary{patterns}
\usepackage{pgfplotstable}
    \pgfplotsset{
    name nodes near coords/.style={
        every node near coord/.append style={
            name=#1-\coordindex,
            alias=#1-last,
        },
    },
    name nodes near coords/.default=coordnode
    }

\usepackage{multirow}

\usepackage[table]{xcolor} 

\newtheorem{theorem}{Theorem}[section]
\newtheorem{definition}[theorem]{Definition}

\newtheorem{lemma}[theorem]{Lemma}
\newtheorem{example}[theorem]{Example}

\newtheorem{proposition}[theorem]{Proposition}

\newcommand{\bX}{\mathbf{X}}
\newcommand{\bB}{\mathbf{B}}
\newcommand{\bP}{\mathbf{P}}

\newcommand{\bc}{\mathbf{c}}

\newcommand{\bw}{\mathbf{w}}

\DeclareMathOperator{\Rank}{Rank}
\DeclareMathOperator{\rk}{rk}

\definecolor{ruleblue}{RGB}{70,120,200}
\definecolor{shadegray}{gray}{0.9}

\makeatletter
\renewcommand{\and}{\end{tabular}\hspace{3em}\begin{tabular}[t]{c}}
\makeatother
\begin{document}
\title{Revisiting EFX and fPO Allocations for Binary Chores}

\author[1]{Zehan Lin}
\author[1]{Xiaowei Wu}
\author[2]{Shengwei Zhou}
\affil[1]{University of Macau, \texttt{\{yc47490,xiaoweiwu\}@um.edu.mo}}
\affil[2]{Nanyang Technological University, \texttt{s.arthur.zhou@gmail.com}}
\date{}

\maketitle

\begin{abstract}
We study fair and efficient allocations of indivisible chores under binary additive costs, with a focus on unequal entitlements and best-of-both-worlds guarantees.
While binary valuations have been extensively studied for goods, their chore counterpart remains much less understood.
A notable exception is due to Tao et al. (TCS 2025), who established that EFX and fPO are compatible for binary chores with equal entitlements.
We revisit binary chores under unequal entitlements and establish the computation of allocations that satisfy WEFX and fPO.
Our key technical contribution is a new lexicographic potential function tailored to chores: any allocation minimizing this potential ensures both WEFX and fPO.
This potential further guides a polynomial-time local-search algorithm for computing allocations with such fairness and efficiency guarantees.
For best-of-both-worlds guarantees, we show that unequal entitlements create a sharp obstruction: even with two agents and identical binary additive costs, no lottery can simultaneously satisfy ex-ante WEF and ex-post WEF1.
In contrast, for equal entitlements, we construct in polynomial time a lottery that is ex-ante EF and fPO, while every realization is both EFX and fPO.
%
\end{abstract}

\maketitle

\section{Introduction}
Fair division, introduced by Steinhaus~\cite{steinhaus1948problem}, is a classic problem in economics and computer science that concerns the allocation of a set $M$ of $m$ items among a set $N$ of $n$ agents with heterogeneous preferences.
When the valuations are positive, the items are viewed as goods (e.g., resources), whereas negative valuations indicate that the items are chores (e.g., tasks).
The objective of fair allocation is to compute an allocation that ensures fairness for agents.
In this paper, we focus on the fair allocation of chores, where each agent $i$ evaluates items via a non-negative and additive cost function $c_i$.

Several notions have been proposed to measure the fairness of allocations.
Among them, envy-freeness (EF)~\cite{foley1966resource} is one of the most well-studied, which requires every agent to weakly prefer her own bundle to every other agent's bundle. 
Unfortunately, EF allocations may fail to exist when items are indivisible.
This motivates relaxations such as envy-freeness up to one item (EF1)~\cite{conf/sigecom/LiptonMMS04} and envy-freeness up to any item (EFX)~\citep{journals/teco/CaragiannisKMPS19}.
Specifically, EF1 requires that any envy between two agents can be eliminated by removing some item from the envious agent’s bundle, whereas EFX requires that envy can be eliminated by removing any item from the envious agent’s bundle.
For chore allocation, the gap between EF1 and EFX is particularly striking: while Bhaskar et al.~\cite{conf/approx/BhaskarSV21} showed that EF1 allocations can be computed in polynomial time, whether EFX allocations always exist had remained a central and technically challenging open problem in fair division until very recently.
A recent breakthrough by He and Tao~\cite{journals/corr/abs-2606-08872} established that EFX allocations may fail to exist even with only $n=4$ agents and three distinct item values, thereby ruling out a general existence guarantee.
Beyond the two-agent case, for which existence follows from the divide-and-choose mechanism, EFX existence has so far been established only in a limited number of highly structured settings, including binary costs~\cite{journals/tcs/TaoWYZ25} and lexicographic preferences~\cite{conf/atal/HosseiniSVX23}.
The existence of EFX allocations remains open for general additive costs even for $n=3$ agents; whereas for bi-valued instances, existence for an arbitrary number of agents remains unresolved, with positive results known only for $n=3$~\cite{journals/ai/ZhouW24, conf/ijcai/GargMQ23} and $n=4$~\cite{journals/corr/abs-2606-08872}.

In this paper, we focus on fair and efficient allocations of binary additive chores, where the cost of each chore is either $0$ or $1$.
Binary valuations admit a natural interpretation as approval votes, a classical model in the voting literature~\cite{books/daglib/0017739,Kilgour10}, and permit simple preference elicitation.
They have therefore received extensive attention in fair division, particularly for the allocation of goods~\cite{conf/ijcai/AleksandrovAGW15,journals/aamas/BouveretL16,conf/aaai/BarmanBMN18,conf/ijcai/FreemanSVX19,conf/aaai/WangW26a}.
Throughout this paper, we use fractional Pareto optimality (fPO) as our efficiency criterion.
An allocation is fPO if no alternative fractional allocation weakly reduces every agent's cost and strictly reduces at least one agent's cost.
Tao et al.~\cite{journals/tcs/TaoWYZ25} established that EFX and fPO are compatible for binary chores.
Building on their result, we revisit the problem with general weights.

\paragraph{Weighted Setting.}
While the traditional fair allocation problem focuses on the case where agents have equal entitlements, real-world settings often involve agents with different levels of responsibility.
The weighted model captures such asymmetry by assigning each agent $i\in N$ a positive weight $w_i$, with $\sum_{i\in N}w_i=1$, representing her entitlement or share of the total obligation; the classical unweighted setting is recovered when $w_i=1/n$ for every agent.
Chakraborty et al.~\cite{journals/teco/ChakrabortyISZ21} introduced \emph{weighted envy-freeness up to one item} (WEF1) for indivisible goods and provided a polynomial-time algorithm for computing such allocations.
Analogous guarantees were later obtained for chores~\cite{conf/aaai/SpringerHY24,journals/ai/WuZZ25}.
However, Springer et al.~\cite{conf/aaai/SpringerHY24} showed that allocations satisfying the stronger notion of \emph{weighted envy-freeness up to any item} (WEFX) may fail to exist for either goods or chores, even with two agents.
That said, the structural restrictions under which WEFX allocations exist remain largely unexplored.
Positive results remain limited to goods: Neoh and Teh~\cite{conf/aaai/NeohT25} showed that a WEFX and fPO allocation can be computed in polynomial time under binary additive valuations, and Liu and Zhang~\cite{journals/corr/abs-2604-08345} extended this result to bi-valued instances.
This motivates our first question on the boundary between the positive and negative cases for chores.

\begin{center}\begin{minipage}{0.95\linewidth}
\noindent\textbf{Question 1.}
\textit{Under what conditions do WEFX allocations for chores exist, and can we characterize the boundary between existence and non-existence?}
\end{minipage}
\end{center}

\paragraph{Best-of-Both-Worlds.}
Although most prior work focuses on deterministic allocations, randomization can offer stronger fairness guarantees in expectation. 
This motivates the study of best-of-both-worlds (BoBW) fairness, which aims to provide guarantees both in expectation (ex-ante) and for every realized allocation (ex-post).
For indivisible goods, Aziz et al.~\cite{journals/ior/AzizFSV24} constructed lotteries that are ex-ante EF and ex-post EF1.
Weighted analogues combining ex-ante WEF with ex-post WEF$(1,1)$ were further obtained for goods~\cite{conf/atal/0001GM23,journals/jair/HoeferSV24} and chores~\cite{journals/ai/WuZZ25}.
Several recent works strengthened the ex-post guarantee to EFX in structured goods settings.
Bu et al.~\cite{conf/wine/BuLLLT24} and Babaioff and Frosh~\cite{journals/corr/abs-2602-14668} obtained ex-ante EF and ex-post EFX for two agents; Bu et al. further extended this guarantee to any number of agents with bi-valued utilities.
Kavitha et al.~\cite{journals/corr/abs-2507-16209} investigated the ex-post EFX guarantee under lexicographic preferences.
These results leave open whether structured chore instances can support the substantially stronger ex-post guarantee of EFX, while retaining exact fairness and efficiency ex-ante.
Binary chores provide a natural test case since EFX and fPO are compatible deterministically~\cite{journals/tcs/TaoWYZ25}.

\begin{center}
    \begin{minipage}{0.95\linewidth}
    \noindent\textbf{Question 2.}
    \textit{For binary chores, can a lottery be fair and efficient ex-ante while every realized allocation is EFX, and how does the answer change under unequal entitlements?}
    \end{minipage}
\end{center}

\subsection{Our Contribution}
In this work, we answer Questions~1 and~2 by studying fair and efficient allocations of binary additive chores, with a focus on unequal entitlements and randomized allocations.
We summarize our main contributions as follows.

\paragraph{WEFX and fPO Allocations.}
We first consider two natural existing approaches: the WEF1 algorithm for chores proposed by Wu et al.~\cite{journals/ai/WuZZ25} and the weighted leximin$^{++}$ framework of Neoh and Teh~\cite{conf/aaai/NeohT25}, which computes WEFX allocations for binary goods. 
Unfortunately, straightforward extensions of both methods fail in our setting. 
Nevertheless, analyzing the underlying reasons for their failure provides some intuition that motivates a novel lexicographic potential function. 
We prove that any allocation minimizing this potential function is both WEFX and fPO. 
Furthermore, its structural properties enable a polynomial-time local search algorithm that resolves WEFX violations while maintaining minimum social cost.
Together, these arguments yield the following result.

\medskip
\begin{center}\begin{minipage}{0.97\linewidth}
\begin{shaded}
\noindent
{\bf Result 1} (Theorem~\ref{theorem: WEFX_PO}){\bf .}
{\em For the allocation of binary additive chores to weighted agents, there exists a polynomial-time
algorithm that computes WEFX and fPO allocations.}
\end{shaded}
\end{minipage}
\end{center}
\medskip

While Springer et al.~\cite{conf/aaai/SpringerHY24} proved that WEFX allocations may fail to exist even for two agents with general additive costs, our result establishes that binary additive costs guarantee both fairness and efficiency for an arbitrary number of agents and highlights this setting as a positive frontier for WEFX.
Moreover, we demonstrate that this existence guarantee cannot be extended along two natural dimensions: relaxing additivity to cancelable binary costs, or broadening the binary domain to bi-valued and restricted additive costs. 
These counterexamples complement our result and fully characterize the boundary between existence and non-existence.

\paragraph{Ex-Ante EF + fPO, and Ex-Post EFX Lotteries.}
To address Question~2, we first show that unequal entitlements present a sharp obstruction: even for two agents with identical binary additive cost functions, no lottery can simultaneously guarantee ex-ante WEF and ex-post WEF1.
We therefore turn to the unweighted setting.
A common approach in the BoBW literature is to begin with an ex-ante fair fractional allocation and decompose it into a lottery satisfying the desired ex-post guarantee~\cite{conf/atal/0001GM23,journals/jair/HoeferSV24,journals/ior/AzizFSV24}.
In contrast, building on Tao et al.~\cite{journals/tcs/TaoWYZ25}, we begin with a deterministic EFX and fPO allocation.
We then design a random reallocation that preserves both properties in every realization.
The main challenge is to choose the reallocation probabilities so that the resulting lottery also satisfies ex-ante EF.
We formulate the admissible reallocations as perfect matchings in a bipartite graph and use a charging argument to establish the feasibility of the ex-ante fairness constraints.
A Birkhoff--von Neumann decomposition then yields the desired lottery in polynomial time.

\medskip
\begin{center}\begin{minipage}{0.97\linewidth}
\begin{shaded}
\noindent
{\bf Result 2} (Theorem~\ref{theorem: Bobw}){\bf .}
{\em For the allocation of binary additive chores to unweighted agents, there exists a polynomial-time
algorithm that computes a lottery that satisfies ex-ante EF, fPO and ex-post EFX, fPO.}
\end{shaded}
\end{minipage}
\end{center}
\medskip

Our result complements the existing literature on BoBW fairness guarantees. 
To the best of our knowledge, few results guarantee ex-post EFX while simultaneously providing non-trivial ex-ante fairness. 
Notable exceptions exist primarily in the goods setting, such as ex-ante EF for bi-valued instances studied by Bu et al.~\cite{conf/wine/BuLLLT24} and ex-ante $9/10$-EF for lexicographic preferences by Kavitha et al.~\cite{journals/corr/abs-2507-16209}.
Finally, our result can be extended to binary cancelable costs to achieve ex-ante EF together with ex-post EFX, where the absence of the efficiency guarantee follows from the known incompatibility between EFX and fPO in this domain~\cite{journals/tcs/TaoWYZ25}.

\subsection{Related Work}

Given the extensive literature on fair division, we focus on work most closely related to binary preferences and the existence of EFX allocations.
For a broader overview, we refer the reader to the surveys~\cite{journals/sigecom/AzizLMW22,journals/ai/AmanatidisABFLMVW23,journals/ipl/Suksompong25}.

\paragraph{Binary Preferences.}
Binary preferences have received substantial attention in fair division~\cite{conf/wine/0002PP020,conf/aaai/BabaioffEF21,journals/mss/SuksompongT22,conf/faw/BuSY23,conf/sagt/BrandlST26}.
Several works have studied welfare rules in this domain, with particular emphasis on the maximum Nash welfare (MNW) rule~\cite{journals/eor/DarmannS15,conf/atal/BarmanKV18,conf/ijcai/0001R20,conf/wine/0002PP020,conf/aaai/BabaioffEF21,conf/sagt/BrandlST26}.
Taken together, these results establish that a suitably tie-broken MNW rule is group-strategyproof and polynomial-time computable, and always returns an EFX and fPO allocation.
For agents with unequal entitlements, Suksompong and Teh~\cite{journals/mss/SuksompongT22} showed that a maximum weighted Nash welfare allocation can be computed in polynomial time under binary valuations.
More recently, Neoh and Teh~\cite{conf/aaai/NeohT25} established the existence of WEFX allocations for binary additive goods.
For binary additive chores, Tao et al.~\cite{journals/tcs/TaoWYZ25} established the existence of EFX and fPO allocations.
Beyond additivity, work remains limited, with two studies focusing on binary supermodular costs~\cite{conf/atal/BarmanNV23,journals/corr/abs-2303-06212}.

\paragraph{EFX Allocations for Chores.}
The existence of EFX allocations was a central open problem for chore allocation until He and Tao~\cite{journals/corr/abs-2606-08872} recently showed that such allocations need not exist under additive costs.
Nevertheless, positive results are known for several important special cases.
Beyond binary costs~\cite{journals/tcs/TaoWYZ25} and lexicographic preferences~\cite{conf/atal/HosseiniSVX23} discussed above, EFX allocations exist for two types of chores~\cite{conf/atal/0001LRS23}, leveled preferences~\cite{journals/teco/GafniHLT23}, and identical ordering preferences~\cite{journals/ai/AzizLMWZ24}.
Garg et al.~\cite{conf/stoc/GargMQ25} and Kobayashi et al.~\cite{journals/tcs/KobayashiMS25} established the existence of EFX allocations when $m \leq 2n$.
Regarding specific cost values, Lin et al.~\cite{conf/ijcai/Lin0025} showed that EFX and fPO allocations can be computed when the cost of each chore is either $1$ or $2$.
Furthermore, Lin et al.~\cite{DBLP:conf/www/LinWZ26} established the existence of EFX allocations for instances where the cost of a chore is either an inherent cost or $0$.


\section{Preliminary}\label{section: preliminary}
We study the problem of fairly allocating a set of $m$ indivisible items (chores), denoted by $M$, to a group of $n$ agents $N$, where each agent $i\in N$ has a weight $w_i > 0$ and $\sum_{i\in N} w_i = 1$.
When $w_i = 1/n$ for all $i\in N$, we call the instance unweighted.
We refer to a subset of items $X \subseteq M$ as a bundle.
Each agent $i \in N$ is associated with a cost function $c_i: 2^M \to \mathbb{R}^+ \cup \{0\}$, which assigns a cost to every bundle of items.
We assume that every cost function is normalized, i.e., $c_i(\emptyset)=0$.
We use $\bc = (c_1, \dots, c_n)$ to denote the cost functions of all agents.
For any $M' \subseteq M$, a partition $\bB = (B_1, \dots, B_n)$ is a collection of disjoint subsets of $M'$ such that $B_i \cap B_j = \emptyset$ for all $i \neq j$ and $\bigcup_{i \in N} B_i = M'$.
An allocation $\bX = (X_1, \dots, X_n)$ is an ordered partition of the item set $M$ into $n$ disjoint bundles, where agent $i$ receives bundle $X_i$.
For any integer $t \geq 1$, we use $[t]$ to denote the set $\{1, \dots, t\}$.
Given an instance $I = (N, M, \bw, \bc)$, our goal is to compute an allocation $\bX$ that is both fair and efficient.
A cost function $c_i$ is additive if $c_i(S) = \sum_{e \in S} c_i(e)$ for any $S \subseteq M$. 
In this paper, we also consider a broader class of cost functions known as cancelable functions, which were introduced by Berger et al.~\cite{conf/aaai/BergerCFF22}.
Formally, a cost function $c_i$ is cancelable if for any two bundles $S, T \subseteq M$ and any item $e \in M \setminus (S \cup T)$, we have
$$
c_i(S \cup \{e\}) > c_i(T \cup \{e\}) \implies c_i(S) > c_i(T).
$$

In this work, we focus on instances where each agent has a binary cost function, which is formally defined as follows.
\begin{definition}[Binary Instances]
    A set $\{c_1, c_2, \dots, c_n\}$ of cost functions is binary if, for every $i\in N$, $S\subseteq M$, and $e\in M\setminus S$, we have $c_i(S\cup\{e\})-c_i(S)\in\{0,1\}$.
\end{definition}

Unless specified otherwise, we assume all binary cost functions to be additive throughout this paper.
For additive cost functions, the binary condition is equivalent to $c_i(e)\in\{0,1\}$ for every item $e$.
Under binary cost functions, we further partition the set of chores as
$$M^+ = \{e \in M : c_i(e) = 1 \text{ for all } i \in N\} \quad \text{and} \quad M^0 = M \setminus M^+.
$$

In the following, we introduce several fairness notions.
	
\begin{definition}[WEF]
    An allocation $\bX$ is weighted envy-free (WEF) if for any $i, j\in N$, 
    \begin{equation*}
        \frac{c_i(X_i)}{w_i} \leq \frac{c_i(X_j)}{w_j}.
    \end{equation*}
\end{definition}
	
Note that when the instance is unweighted, the notion of WEF coincides with the EF notion.
Hence, WEF allocations are not guaranteed to exist.
To address this, we next introduce two relaxed fairness notions derived from WEF.
	
\begin{definition}[WEF1]
    An allocation $\bX$ is weighted envy-free up to one item (WEF1) if for any agents $i, j\in N$, either $X_i = \emptyset$, or there exists an item $e\in X_i$ such that
    \begin{equation*}
        \frac{c_i(X_i \setminus \{e\})}{w_i} \leq \frac{c_i(X_j)}{w_j}.
    \end{equation*}
\end{definition}

\begin{definition}[WEFX]
    An allocation $\bX$ is weighted envy-free up to any item (WEFX) if for any agents $i, j\in N$, either $X_i = \emptyset$, or for every item $e\in X_i$, we have
    \begin{equation*}
        \frac{c_i(X_i \setminus \{e\})}{w_i} \leq \frac{c_i(X_j)}{w_j}.
    \end{equation*}
\end{definition}
For convenience of
notation, we use $\hat{c}_i(X_i)$ to denote the cost for agent $i$ after moving the item with minimum cost, i.e., 
$$
\hat{c}_i(X_i) = \max_{e \in X_i} c_i(X_i \setminus \{e\}),
$$
with the convention $\hat{c}_i(\emptyset)=0$.
Next, we formally define the efficiency notions considered throughout this paper.

\begin{definition}
     The social cost of an allocation $\textbf{X}$ is defined as 
     $sc(\textbf{X}) = \sum_{i \in N}c_i(X_i)$.
\end{definition}

\begin{definition}[PO and fPO]
    An allocation $\bX'$ \emph{Pareto dominates} another allocation $\bX$ if $c_i(X_i') \leq c_i(X_i)$ for all $i \in N$, with strict inequality holds for at least one agent. 
    An integral allocation $\bX$ is \emph{Pareto optimal} (PO) if it is not Pareto dominated by any other integral allocation. 
    Furthermore, $\bX$ is \emph{fractional Pareto optimal} (fPO) if it is not Pareto dominated by any other fractional allocation\footnote{In a fractional allocation, items can be fractionally assigned. Formally, it is represented by variables $x_{ij} \in [0, 1]$ denoting the fraction of item $j$ allocated to agent $i$, such that $\sum_{i \in N} x_{ij} = 1$ for every item $j$. The cost incurred by an agent is extended linearly, i.e., $c_i(X_i) = \sum_{j} x_{ij} c_{ij}$.}.
\end{definition}

\section{WEFX Allocations for Binary Costs}
In this section, we study the computation of WEFX allocations for binary additive costs.
We introduce a novel lexicographic potential function in Section~\ref{subsection: potential} and show that minimizing this potential function guarantees both WEFX and fPO allocations.
Guided by the structure of the potential function, we design a polynomial-time algorithm in Section~\ref{subsection: WEFXPO} that efficiently computes a WEFX and PO allocation. 
Finally, Section~\ref{subsection: extension} explores the boundaries of WEFX existence by considering more general settings, including cancelable binary, bi-valued, and restricted additive cost functions.
To complement our result, we also discuss existing algorithms and explain why they fail to directly extend to our setting in Appendix~\ref{appendix:existing-algorithms}.
\paragraph{High-Level Idea.}
To establish the existence of a WEFX allocation for chores, a natural approach is to design a potential function and show that minimizing it directly implies a WEFX allocation, as Neoh and Teh~\cite{conf/aaai/NeohT25} demonstrated in the goods setting.
However, directly following their approach by lexicographically minimizing the normalized cost $c_i(X_i)/w_i$ is not sufficient to guarantee WEFX for chores (see Appendix~\ref{appendix:existing-algorithms} for details).
This discrepancy stems from the fundamental difference between the goods and chores settings.
For each agent $i \in N$, the cost $\hat{c}_i(X_i)$ is used for comparison in the chores setting, whereas the direct valuation $v_i(X_i)$ is used for goods.
Hence, a more reasonable approach is to lexicographically minimize the value of $\hat{c}_i(X_i)/w_i$ for $i \in N$, as it naturally captures whether an agent receives any zero-cost items.

\subsection{A Lexicographic Potential Function} \label{subsection: potential}

In this subsection, we define the potential function and demonstrate that any allocation minimizing it is both WEFX and fPO.
We first introduce some necessary notation. Based on the composition of their bundles, we partition the agents into two groups:
$$ 
N_0 = \{i \in N : c_i(X_i) \neq |X_i| \} \quad \text{and} \quad N_1 = N \setminus N_0. 
$$
Note that an agent's membership in $N_0$ or $N_1$ may change dynamically as items are transferred between agents. 
We define the primary potential $\Phi(\bX)$ as the vector of $(\hat{c}_i(X_i) / w_i)_{i \in N}$ sorted in non-increasing order, complemented by the tie-breaking potential $\Psi(\bX)$:
$$
\Phi(\bX)=\operatorname{sort}^{\downarrow}\left(\left(\frac{\hat{c}_i(X_i)}{w_i}\right)_{i\in N}\right),
\qquad
\Psi(\bX) = \sum_{i \in N_0} \frac{c_i(X_i)}{w_i}.
$$
The coordinates of $\Phi(\bX)$ are compared lexicographically.
Ultimately, we define the composite potential function $\mathcal{P}(\bX)$ as the following tuple, which is also evaluated lexicographically:
$$\mathcal{P}(\bX) = \left( sc(\bX), \Phi(\bX), -\Psi(\bX) \right).
$$

Recall that $M^+$ contains the items that have a cost of $1$ for every agent.
Therefore, any allocation has a social cost of at least $|M^+|$.
This implies the following lemma immediately.

\begin{lemma} \label{lemma: fpo}
    If an allocation $\bX$ satisfies $sc(\bX) = |M^+|$, then $\bX$ is fPO.
\end{lemma}


Next, we show the main result of this subsection. 

\begin{lemma} \label{lemma: potential_power}
    If an allocation $\bX$ lexicographically minimizes $\mathcal{P}$, then $\bX$ is WEFX and fPO.

\end{lemma}
\begin{proof}
By assigning each chore in $M^0$ to an agent who incurs zero cost and allocating the chores in $M^+$ arbitrarily, we minimize the social cost to $|M^+|$.
Since $\bX$ lexicographically minimizes $\mathcal{P}$, its first coordinate must satisfy $sc(\bX)=|M^+|$, which implies that $\bX$ is fPO by Lemma~\ref{lemma: fpo}.

Therefore, it remains to prove that $\bX$ is WEFX.
Suppose, for contradiction, that $\bX$ is not WEFX.
Then, there exist agents $i, j \in N$ such that
\begin{equation}
    \frac{\hat{c}_i(X_i)}{w_i}>\frac{c_i(X_j)}{w_j} \ge 0.
    \label{equation:not_WEFX}
\end{equation}

The violating inequality implies $\hat{c}_i(X_i)>0$.
Given that $sc(\bX)=|M^+|$, all items in $M^0$ must have been assigned to agents who incur zero cost for them.
Consequently, the positive cost incurred by agent $i$ must be caused by items from $M^+$, which implies $X_i \cap M^+ \neq \emptyset$. 
Let $S = X_j \cap M^0$ be the zero-cost items allocated to agent $j$.
We consider two cases.

\paragraph{Case 1: $S=\emptyset$, or there exists some $e'\in S$ such that $c_i(e')=1$.}
In this case, we construct a new allocation $\mathbf{Y}$ by transferring an item $e \in X_i \cap M^+$ from agent $i$ to agent $j$:
$$
Y_i = X_i \setminus \{e\}, \qquad Y_j = X_j \cup \{e\}, \qquad Y_k = X_k \text{ for every } k \notin \{i,j\}.
$$
Since $e \in M^+$, this transfer preserves the total social cost, i.e., $sc(\mathbf{Y}) = sc(\bX)$, which maintains fPO by Lemma~\ref{lemma: fpo}.
Removing $e$ from $X_i$ yields $\hat{c}_i(Y_i) = \hat{c}_i(X_i) - 1$, which implies that $\hat{c}_i(Y_i)/w_i < \hat{c}_i(X_i)/w_i$.
We next show that $\hat{c}_j(X_j)/w_j$ and $\hat{c}_j(Y_j)/w_j$ are also smaller than $\hat{c}_i(X_i)/w_i$.
\begin{itemize}
    \item If $S=\emptyset$, every item in $X_j$ belongs to $M^+$, and hence $c_i(X_j)=c_j(X_j)$. Therefore,
    $$
    \frac{\hat{c}_j(X_j)}{w_j}\leq \frac{\hat{c}_j(Y_j)}{w_j}=\frac{c_j(X_j)}{w_j}=\frac{c_i(X_j)}{w_j}<\frac{\hat{c}_i(X_i)}{w_i}.
    $$
    \item Otherwise, $j\in N_0$. Since some $e'\in S$ costs one to agent $i$, while every item in $M^+$ costs one to both agents, we have $c_i(X_j)\geq c_j(X_j)+1$. Therefore,
    $$
    \frac{\hat{c}_j(X_j)}{w_j}<\frac{\hat{c}_j(Y_j)}{w_j}=\frac{c_j(X_j)+1}{w_j}\leq\frac{c_i(X_j)}{w_j}<\frac{\hat{c}_i(X_i)}{w_i}.
    $$
\end{itemize}
Therefore, $\Phi(\mathbf{Y})$ is lexicographically smaller than $\Phi(\bX)$.
Combined with $sc(\mathbf{Y}) = sc(\bX)$, this contradicts the minimality of $\mathcal{P}(\bX)$.

\paragraph{Case 2: $S\neq\emptyset$ and $c_i(e')=0$ for every $e'\in S$.}
In this case, we construct a new allocation $\mathbf{Y}$ by moving $e$ from agent $i$ to agent $j$ while shifting the entire set $S$ from agent $j$ to agent $i$:
$$
Y_i = (X_i \setminus \{e\}) \cup S, \qquad Y_j = (X_j \cup \{e\}) \setminus S, \qquad Y_k = X_k \text{ for every } k \notin \{i,j\}.
$$
Note that $S$ incurs zero cost for both agents $i$ and $j$.
Since $e\in M^+$, this exchange preserves the social cost, and we have $sc(\bX) = sc(\mathbf{Y})$.
Before the exchange, $j\in N_0$; after losing all chores in $S=X_j\cap M^0$, agent $j$ belongs to $N_1$.
Therefore, we have $\hat{c}_j(Y_j)=c_j(X_j)=\hat{c}_j(X_j)$.
Agent $i$ receives a zero-cost chore and therefore belongs to $N_0$ after the exchange, while her own cost decreases by one.
\begin{itemize}
    \item If $i\in N_0$ before the exchange, then $\hat{c}_i(Y_i)=\hat{c}_i(X_i)-1$. Hence, $\Phi(\mathbf{Y})$ is lexicographically smaller than $\Phi(\bX)$, and so is $\mathcal{P}(\mathbf{Y})$.

    \item If $i\in N_1$ before the exchange, then $\hat{c}_i(Y_i)=c_i(X_i)-1=\hat{c}_i(X_i)$, which implies $\Phi(\mathbf{Y})=\Phi(\bX)$.
    Moreover, as $i$ enters $N_0$ and $j$ leaves $N_0$, we have
    $$
    \Psi(\mathbf{Y})-\Psi(\bX)
    = \frac{c_i(X_i)-1}{w_i}-\frac{c_j(X_j)}{w_j}
    = \frac{\hat{c}_i(X_i)}{w_i}-\frac{c_j(X_j)}{w_j} 
    = \frac{\hat{c}_i(X_i)}{w_i}-\frac{c_i(X_j)}{w_j} 
    > 0,
    $$
    where the last equality holds because $c_i(S) = c_j(S) = 0$ and the last inequality holds by Equation~\eqref{equation:not_WEFX}.
    Therefore, the first two coordinates of $\mathcal{P}$ remain unchanged while its third coordinate $-\Psi$ strictly decreases, again contradicting the minimality of $\bX$.
\end{itemize}

Consequently, in both cases, a WEFX violation yields an allocation with a lexicographically smaller potential.
Therefore, $\bX$ is WEFX and fPO.
\end{proof}
\subsection{A Polynomial-Time Algorithm}\label{subsection: WEFXPO}
Observe that the previous potential function reveals the underlying structure for computing WEFX allocations. 
Unfortunately, merely minimizing the potential function $\mathcal{P}$ only guarantees existence. 
In this subsection, we design a local search algorithm inspired by the potential function and show that by carefully transferring items, the algorithm (Algorithm~\ref{alg: WEFX}) will terminate in polynomial time.
We first provide an overview of Algorithm~\ref{alg: WEFX}.

\paragraph{Algorithm Overview.}
Algorithm~\ref{alg: WEFX} begins by partitioning the chores into two sets, $M^+$ and $M^0$ as in the previous definition. 
Each chore in $M^0$ is assigned to an agent who incurs a cost of zero for it, while the chores in $M^+$ are initially allocated arbitrarily. 
Starting from this initial allocation, the algorithm iteratively resolves WEFX violations.
Specifically, whenever the algorithm identifies a pair of agents $i, j \in N$ under the current allocation $\bX$ such that $\hat{c}_i(X_i)/w_i > c_i(X_j)/w_j$, it performs local reallocations until this violation is successfully resolved.

\begin{algorithm}[!ht]
\caption{Computation of WEFX and fPO Allocation for Binary Chores}
\label{alg: WEFX}
\KwIn{Weighted binary instance $I = (N, M, \{c_i\}_{i \in N}, \{w_i\}_{i \in N})$}
    Initialize $X_i \gets \emptyset$ for $i \in N$; \\
    Let $M^+ = \{e \in M: c_i(e) = 1 \text{ for all } i \in N\}$, $M^0 = M \setminus M^+$; \\
    \For{$e \in M^0$}{
        Choose any agent $i \in N$ such that $c_i(e) = 0$; \\
        $X_i \gets X_i \cup \{e\}$; \\
    }
    \For{$e \in M^+$}{
        Choose any agent $i \in N$; \\
        $X_i \gets X_i \cup \{e\}$; \\
    }
    \While{there exist agents $i$ and $j$ such that $\hat{c}_i(X_i)/w_i > c_i(X_j)/w_j$}{
        Pick any $e \in X_i \cap M^+$; \\
        $S \gets X_j \cap M^0$; \\
        \If{$S = \emptyset$ or there exists $e' \in S$ with $c_i(e') = 1$}{
            $X_i \gets X_i \setminus \{e\}$; \\
            $X_j \gets X_j \cup \{e\}$; \\
        } \Else {
            $X_i \gets (X_i \setminus \{e\}) \cup S$; \\
            $X_j \gets (X_j \cup \{e\}) \setminus S$; \\
        }
    }
\KwOut{A WEFX and fPO allocation $\bX$}
\end{algorithm}

\medskip

It is straightforward to verify that throughout the whole algorithm, every item $e\in M^0$ is allocated to an agent with zero cost for it.
Therefore, we have the following immediately.

\begin{lemma} \label{lemma: fpo_poly}
Throughout the execution of Algorithm~\ref{alg: WEFX}, every chore $e \in M^0$ is assigned to an agent $i \in N$ satisfying $c_i(e)=0$.
\end{lemma}

Given that the WEFX property is immediately satisfied when Algorithm~\ref{alg: WEFX} terminates, we only need to show that it runs in polynomial time.

\begin{lemma} \label{lemma: polynomial}
    Algorithm~\ref{alg: WEFX} terminates in polynomial time.
\end{lemma}

\begin{proof}
If $M^+=\emptyset$, the algorithm only allocates items in $M^0$ and terminates immediately.
Therefore, we assume that $M^+\neq\emptyset$ in the following.
Let $\bX$ and $\bX'$ denote the allocations before and after an iteration, respectively.
The while-condition implies $\hat{c}_i(X_i)>0$, which further implies $X_i \cap M^+ \neq \emptyset$ by Lemma~\ref{lemma: fpo_poly}.
Next, we introduce some notations. 

Consider all pairs $(i,t)$ where $i\in N$ and $t \in \{1,\ldots,\vert{}M^+\vert{}\}$. 
We sort these pairs in non-decreasing order based on the value of $t/w_i$, breaking ties in favor of the smaller agent index. 
Let $\rk(i,t)$ denote the position (rank) of the pair $(i,t)$ in this total ordering, which is an integer in $[n\cdot |M^+|]$.
Note that by definition, the minimum rank $1$ corresponds to the pair achieving $\min_{i \in N} \{1/w_i\} = 1/\max_{i \in N} \{w_i\}$, while the maximum rank $n\cdot |M^+|$ corresponds to the pair achieving $|M^+| / \min_{i \in N} \{w_i\}$. 
We define the overall rank of an allocation $\bX$ as follows:
$$
\Rank(\bX)
=\sum_{i\in N}\sum_{t=1}^{\hat{c}_i(X_i)}\rk(i,t).
$$
Note that we have $\Rank(\bX) \leq n^2 \cdot |M^+|^2 \leq n^2m^2$.

In the following, we analyze how an iteration changes $\Rank(\bX)$.
Suppose first that the \texttt{if} branch is executed.
Removing $e$ from $X_i$ gives $\hat{c}_i(X'_i)=\hat{c}_i(X_i)-1$, so the term $\rk(i,\hat{c}_i(X_i))$ is removed from $\Rank(\bX)$.
If $S=\emptyset$, then $X_j$ contains only chores in $M^+$, so $\hat{c}_j(X'_j)=c_j(X_j)=c_i(X_j)$.
Otherwise, there exists $e'\in S$ (note that $S = X_j \cap M^0$) such that $c_i(e') = 1$, and hence 
\begin{equation*}
    \hat{c}_j(X'_j) = c_j(X'_j) = c_j(X_j \cup \{e\}) = c_j(X_j)+1 \leq c_i(X_j).
\end{equation*}
In either case, the while-condition gives $\hat{c}_j(X'_j)/w_j\leq c_i(X_j)/w_j<\hat{c}_i(X_i)/w_i$.
Depending on the contents of $X_j$, we consider the following two cases:
\begin{itemize}
    \item If $X_j\neq\emptyset$, since it holds $\hat{c}_j(X'_j)=\hat{c}_j(X_j)+1$, the term $\rk(j,\hat{c}_j(X'_j))$ is added to overall rank.
    It follows that 
    \begin{equation*}
        \Rank(\bX') - \Rank(\bX) = \rk(j, \hat c_j(X'_j)) - \rk(i, \hat c_i(X_i)) < 0.
    \end{equation*}
    \item Otherwise $X_j=\emptyset$, then $\hat{c}_j(X'_j)=0$ and no term is added.
    It simply follows that 
    $$\Rank(\bX') - \Rank(\bX) = - \rk(i, \hat c_i(X_i))<0.
    $$
\end{itemize}
Consequently, we have $\Rank(\bX')<\Rank(\bX)$ in this branch.

It remains to consider that the \texttt{else} branch is executed.
It follows that $S \neq \emptyset$, and $c_i(S) = c_j(S) = 0$. 
Since $X'_j = (X_j \setminus S) \cup \{e\}$, we have $\hat{c}_j(X'_j) = c_j(X_j \setminus S) = c_j(X_j) = \hat{c}_j(X_j)$. 
Hence the change of overall rank only depends on those related to agent $i$.
We distinguish two cases based on agent $i$'s initial status.
\begin{itemize}
    \item If $i \in N_0$ prior to the update, then $\hat{c}_i(X'_i) = \hat{c}_i(X_i) - 1$. Consequently, the term $\rk(i, \hat{c}_i(X_i))$ is eliminated from the summation, yielding $\Rank(\bX') < \Rank(\bX)$.
    \item If $i \in N_1$, then $i$ enters $N_0$ and $\hat{c}_i(X'_i) = c_i(X'_i) = c_i(X_i) - 1 = \hat{c}_i(X_i)$. Under this scenario, we have $\Rank(\bX') = \Rank(\bX)$.
\end{itemize}
An iteration falling into the latter case is referred to as a \emph{swap step}. 
Notably, any iteration that is not a swap step (which we define as a \emph{non-swap step}) strictly reduces $\Rank(\bX)$.
Given that $\Rank(\bX)$ always takes integer values, and that each non-swap step strictly decreases this rank by at least $1$ while each swap step leaves it unchanged, the algorithm can perform at most $n^2m^2$ non-swap steps throughout its entire execution.

It remains to bound the number of swap steps.
Fix a maximal consecutive sequence of swap steps.
The preceding analysis shows that every swap step preserves $\hat c_i(X_i)/w_i$ for every agent $i\in N$.
Hence, throughout this sequence, the agents have a fixed total ordering by non-decreasing $\hat c_i(X_i)/w_i$, with ties broken in favor of smaller agent indices.
Let $p(i)\in[n]$ denote the position of agent $i$ in this ordering, and define a potential function $f(\bX)=\sum_{i\in N_0}p(i)$.

Consider a swap step involving a violating pair of agents $i$ and $j$, and let $\bX$ and $\bX'$ denote the allocations before and after the step, respectively.
Before the step, $i\in N_1$ and $j\in N_0$, whereas afterward, $i$ enters $N_0$ and $j$ leaves $N_0$.
Moreover, the \texttt{else} branch implies that $S=X_j\cap M^0$ is nonempty and costs zero to agent $i$.
By Lemma~\ref{lemma: fpo_poly}, every chore in $S$ also costs zero to agent $j$.
Since every chore in $X_j\setminus S$ belongs to $M^+$, we obtain $c_i(X_j)=c_j(X_j)=\hat{c}_j(X_j)$.
Together with the while-condition, this gives
$$
    \frac{\hat{c}_i(X_i)}{w_i}
    >
    \frac{c_i(X_j)}{w_j}
    =
    \frac{\hat{c}_j(X_j)}{w_j}.
$$
Therefore, we have $p(i)>p(j)$, and hence
$$
    f(\bX')-f(\bX)
    =
    p(i)-p(j)
    \geq 1.
$$
Since the potential function $f$ is integer-valued and satisfies $f(\bX)\leq\sum_{r=1}^n r= n(n+1)/2$, every maximal consecutive sequence contains at most $n(n+1)/2$ swap steps.
Since any two maximal sequences of swap steps are separated by a non-swap step, the total number of steps is at most $n(n+1)(n^2 m^2+1)/2$, which is polynomial in $n$ and $m$.
Furthermore, each iteration can be implemented in polynomial time, as it requires only enumerating all pairs of agents and scanning their bundles.
Hence, Algorithm~\ref{alg: WEFX} terminates in polynomial time.
\end{proof}

Consequently, we can establish the following theorem.
\begin{theorem} \label{theorem: WEFX_PO}
    For the allocation of binary additive chores to weighted agents, Algorithm~\ref{alg: WEFX} computes WEFX and fPO allocations in polynomial time.
\end{theorem}

\subsection{Non-existence of WEFX in More General Settings}\label{subsection: extension}
While Theorem~\ref{theorem: WEFX_PO} guarantees the existence of WEFX allocations in the weighted setting, a natural subsequent step is to characterize the tight boundaries of this existence.
In the following, we explore three natural relaxations of our setting. 
Specifically, on the one hand, we still assume binary costs but replace additivity with cancelability. 
On the other hand, we maintain additivity while extending the cost domain to bi-valued or restricted additive costs.
Unfortunately, we show that WEFX allocations need not exist in each of these directions.
\begin{example}[Non-existence of WEFX for cancelable binary instances]\label{example:cancelable-nonexistence}
Consider two agents with weights $w_1=4/9$ and $w_2=5/9$, and $M=\{e_1,e_2,e_3,e_4\}$.
The cost function of each agent for every subset $S\subseteq M$ is given in Table~\ref{tab:cancelable-nonexistence}.
Although both cost functions are binary and cancelable, no WEFX allocation exists for this instance.
\end{example}

\begin{table}[!htbp]
\centering
\caption{Cost functions for the binary cancelable counterexample}
\label{tab:cancelable-nonexistence}
\upshape
\begin{tabular}{ccc}
\toprule
& $S = \emptyset$ & $S\neq\emptyset$ \\
\midrule
Agent 1 & $0$ & $1$ \\
Agent 2 & $0$ & $|S|$ \\
\bottomrule
\end{tabular}
\end{table}

We first show that $c_1$ and $c_2$ are binary cancelable cost functions.
Since all marginal costs are easily seen to be either $0$ or $1$, it suffices to verify cancelability.
To this end, consider any two bundles $S, T \subseteq M$ and any chore $e \notin S \cup T$.
\begin{itemize}
    \item For $c_1$, adding $e$ yields a non-empty bundle with a cost of $1$, so $c_1(S \cup \{e\}) > c_1(T \cup \{e\})$ never holds, making the condition trivially satisfied.
    \item For $c_2$, since $c_2(S) = |S|$, the inequality $c_2(S \cup \{e\}) > c_2(T \cup \{e\})$ immediately implies $c_2(S) > c_2(T)$ by $|S| > |T|$.
\end{itemize}
Hence, both $c_1$ and $c_2$ are binary cancelable cost functions.

We now show that no allocation $(X_1, X_2)$ in this instance satisfies WEFX.
If $|X_1| \in \{0, 4\}$, there exists one agent that receives all chores, which trivially violates WEFX.
Otherwise, we have $|X_1| \in \{1, 2, 3\}$ and both agents receive non-empty bundles:
\begin{itemize}
\item If $|X_1| = 1$, from the perspective of agent $2$, we have
$$
\frac{\hat c_2(X_2)}{w_2} = \frac{18}{5} > \frac{9}{4} = \frac{c_2(X_1)}{w_1}.
$$
\item Otherwise, $|X_1| \in \{2, 3\}$, agent $1$ violates the condition of WEFX as 
$$
\frac{\hat c_1(X_1)}{w_1} = \frac{9}{4} > \frac{9}{5} = \frac{c_1(X_2)}{w_2}.
$$
\end{itemize}
In all cases, WEFX is violated, which completes the proof.
\medskip

We next enlarge the binary domain into two distinct directions: bi-valued instances and restricted additive instances. 
The counterexamples and non-existence proofs for both settings are deferred to Appendix~\ref{appendix:wefx-nonexistence}.
Together with the cancelable binary instance above, these results show that the guarantee of Theorem~\ref{theorem: WEFX_PO} relies jointly on binary costs and additivity.

\section{BoBW Guarantees for Binary Chores}
In this section, we shift our focus to randomized allocations and investigate the existence of BoBW guarantees. 
We show that the feasibility of BoBW properties depends heavily on whether agents have equal entitlements. Specifically, we first demonstrate that under unequal weights, ex-ante WEF is incompatible with ex-post WEF1. 
This impossibility naturally motivates our exploration of the unweighted setting, where we successfully establish the coexistence of ex-ante EF and fPO with ex-post EFX and fPO for binary additive costs (Theorem~\ref{theorem: Bobw}).  
Finally, we extend these fairness properties to the broader domain of binary cancelable costs, where ex-ante EF and ex-post EFX can still be successfully guaranteed (Theorem~\ref{theorem: cancelable}).

We first show the incompatibility of the BoBW guarantee in the weighted setting.
Consider an instance with two agents, $N = \{1, 2\}$, and two chores, $M = \{e_1, e_2\}$, where $w_1 > w_2$ and $w_1 + w_2 = 1$. 
Assume both agents have identical binary additive costs, where every chore has a cost of $1$ for both agents, as shown in Table \ref{tab:wefx_impossibility}.
\begin{table}[ht]
\centering
\caption{Non-existence of ex-ante WEF and ex-post WEF1}
\label{tab:wefx_impossibility}
\begin{tabular}{lcc}
\toprule
 & $e_1$ & $e_2$ \\
\midrule
Agent 1  & $1$ & $1$ \\
Agent 2  & $1$ & $1$ \\
\bottomrule
\end{tabular}
\end{table}

Suppose, for the sake of contradiction, that there exists a randomized allocation that is ex-ante WEF and ex-post WEF1.
To satisfy ex-post WEF1, every realization must allocate exactly one item to each agent.
Consequently, we have
$$
\frac{\mathbb{E}[c_2(X_2)]}{w_2} = \frac{1}{w_2} >  \frac{1}{w_1} = \frac{\mathbb{E}[c_2(X_1)]}{w_1},
$$
which implies that agent $2$ envies agent $1$ in expectation, a contradiction to ex-ante WEF.
Therefore, no randomized allocation can guarantee ex-ante WEF and ex-post WEF1.

Given the strong incompatibility result, in the following, we focus on the unweighted setting.

\subsection{Feasible Configurations and Ex-post Guarantees}

Intuitively, to ensure an EFX allocation, our idea is to evenly partition the items in $M^+$. 
We assume that $|M^+| = sn + r$, where $0 \le r < n$.
Specifically, we select a subset of agents $N_H \subseteq N$ (called \emph{high agents}) with $|N_H| = r$, each of whom receives $s + 1$ chores from $M^+$. The remaining agents (called \emph{low agents}) receive exactly $s$ chores from $M^+$.
The induced envy from a high agent to a low agent will be compensated by a \emph{feasible configuration} involving items in $M^0$.
Below, we formally define the notion of a feasible configuration.
\begin{definition}[Feasible Configuration]
    A pair $(\bP, N_H)$ is a feasible configuration if:
    \begin{itemize}
        \item $\bP=(P_1, \dots, P_n)$ is a partition of $M^0$, where $c_i(P_i) = 0$ for every $i \in N$;
        \item $N_H \subseteq N$ is the set of high agents and $|N_H| = r$;
        \item for every agent $i \in N_H$ with $P_i \neq \emptyset$ and every agent $j \notin N_H$, we have $c_i(P_j) \geq 1$. 
    \end{itemize}
\end{definition}

Note that in the remainder of this section, we focus exclusively on feasible configurations, i.e., we always maintain a balanced partition of $M^+$ and distinguish between high agents (who receive $s + 1$ chores) and low agents (who receive $s$ chores).
We first show constructively that a feasible configuration always exists.

\begin{lemma}\label{lemma: feasible_configuration}
    A feasible configuration can be computed in polynomial time.
\end{lemma}
\begin{proof}
Fix an arbitrary ordering of the agents.
For each chore $e \in M^0$, allocate $e$ to the first agent $i$ in this ordering satisfying $c_i(e)=0$; the existence of such an agent follows from the definition of $M^0$.
Let $\bP=(P_1,\ldots,P_n)$ denote the resulting partition, and let $N_1=\{i\in N:P_i=\emptyset\}$ be the set of agents who do not receive any item from $M^0$.

If $|N_1|\geq r$, let $N_H$ consist of any $r$ agents in $N_1$.
The second and third conditions in the definition of a feasible configuration then hold immediately.

If $|N_1|<r$, let $N_H$ contain all agents in $N_1$, together with the first $r-|N_1|$ agents in the ordering who have nonempty bundles.
To verify feasibility, consider any agent $i\in N_H$ with $P_i\neq\emptyset$ and any agent $j\notin N_H$.
Since $N_1\subseteq N_H$, we have $P_j\neq\emptyset$.
Moreover, agent $i$ precedes agent $j$ in the ordering.
For every chore $e\in P_j$, agent $j$ is the first agent in the ordering who incurs zero cost for $e$.
It follows that $c_i(e)=1$, and hence $c_i(P_j)\geq1$.

Therefore, $(\bP,N_H)$ is a feasible configuration.
The construction scans the agents only once for each chore and thus runs in polynomial time.
\end{proof}

Given any feasible configuration $(\bP, N_H)$, we can construct an allocation $\bX$ by assigning $P_i$ to every agent $i$ and then arbitrarily partitioning $M^+$ so that each high agent receives $s+1$ chores and each low agent receives $s$ chores.
Below, we demonstrate that the resulting allocation $\bX$ is always EFX and fPO.

\begin{lemma} \label{lemma: expost_EFX}
    Any integral allocation $\bX$ induced by a feasible configuration is EFX and fPO.
\end{lemma}
\begin{proof}
The allocation of items in $M^+$ guarantees that $c_i(X_j) \geq s$ for all $i,j\in N$;
the allocations of items in $M^0$ guarantees that $c_i(X_i) \leq s+1$ for all $i\in N$.
Therefore, if agent $i$ is not EFX towards agent $j$, then it must be the case that $c_i(X_i) = s+1$, $c_i(X_j) = s$, and we further have $P_i \neq \emptyset$.
However, the definition of feasible configuration rules out such a case.
Therefore, the allocation $\bX$ is EFX.
%
%
Finally, since $sc(\bX)=|M^+|$, Lemma~\ref{lemma: fpo} implies that $\bX$ is fPO.
\end{proof}

\subsection{Reallocation via Perfect Matchings}
\label{subsection: matching_reallocation}

Given the EFX guarantee provided by a feasible configuration, to further ensure ex-ante EF, we construct a polynomial-time lottery from the perspective of matching.
Fix the feasible configuration $(\bP,N_H)$ returned by Lemma~\ref{lemma: feasible_configuration}. 
Let $N_H^+=\{i\in N_H:P_i\neq\emptyset\}$ be the high agents who receive some item from $M^0$.
Note that agents in $N_H^+$ do not envy and are not envied by any other agents:
\begin{itemize}
    \item For all $j\in N$ and $i\in N_H^+$, we have $c_j(X_j)\leq s+1 \leq c_j(X_i)$.
    \item By definition of feasible configuration, for all $i\in N_H^+$ and $j\notin N_H$, we have $c_i(P_j)\geq 1$. Therefore, $c_i(X_i) = s+1 \leq c_i(X_j)$.
\end{itemize}

Therefore, we keep agents in $N_H^+$ and their bundles fixed and reassign the other bundles among the remaining agents $N'=N\setminus N_H^+$. 
We construct a bipartite graph with all agents in $N'$ on the left. 
The right side also contains $|N'|$ nodes, which consist of $n-r$ bundle nodes $P_k$ for every low agent $k \notin N_H$ and $q = r - |N_H^+|$ high slots.
Next, we construct the edges.

For each $i \in N'$ and $k \notin N_H$, the graph contains the edge $(i, P_k)$ if and only if $c_i(P_k) = 0$. 
Matching agent $i$ to $P_k$ assigns $P_k$ to her and makes her a low agent who receives $P_k$ along with $s$ chores from $M^+$. 
Additionally, every agent in $N'$ is adjacent to all high slots.
Matching an agent to a high slot assigns her no chores from $M^0$ and lets her be a high agent who receives exactly $s+1$ chores from $M^+$.
Intuitively, agents in $N'$ can switch their roles depending on the chosen matching.
%
%
%
%
Given a perfect matching $T$ in this graph, let $(\tilde \bP, \tilde N_H)$ denote the induced configuration.
Each agent $i \in N_H^+$ keeps their original bundle. 
For each agent $i \in N'$,
\begin{itemize}
    \item if $i$ is matched to a bundle vertex $P_k$, we set $\tilde P_i = P_k$ and $i \notin \tilde N_H$;
    \item if $i$ is matched to a high slot, we set $\tilde P_i = \emptyset$ and $i \in \tilde N_H$.
\end{itemize}

In the following, we show that any perfect matching induces a feasible configuration.

\begin{lemma}\label{lemma: matching_configuration}
    The graph always contains a perfect matching, where every such matching induces a feasible configuration.
\end{lemma}
\begin{proof}
    We first show the existence of a perfect matching. 
    We can match each bundle vertex $P_k$ (where $k \notin N_H$) to agent $k$, which is valid since $c_k(P_k) = 0$. The remaining $q$ high slots are then matched arbitrarily to the agents in $N_H \setminus N_H^+$, with each agent receiving exactly one slot.
    
    Now, fix any perfect matching $T$ in the graph. 
    Let $(\tilde \bP, \tilde N_H)$ denote the configuration induced by $T$, where $\tilde P_i$ denotes the bundle of items from $M^0$ allocated to agent $i$.
    %
    Clearly, by construction we have $c_i(\tilde P_i) = 0$ for all $i\in N$.
    Furthermore, the construction of $T$ yields $|\tilde N_H| = |N_H^+| + q = |N_H^+| + r - |N_H^+| = r$.
    Finally, by the feasibility of the initial configuration $(\bP, N_H)$, we have $c_i(P_k) \geq 1$ for all $i\in N_H^+$ and $k\notin N_H$.
    Since each low agent $j$ is matched to an original bundle $P_k$ for some $k \notin N_H$, the feasibility of $(\tilde \bP, \tilde N_H)$ follows.
\end{proof}

\subsection{A Linear Program for Ex-ante Envy-freeness}
\label{subsection: matching_lp}

In the following, we introduce a linear program (LP) for fractional matchings.
Specifically, for $i\in N'$ and $k\notin N_H$, let $x_{i,k}$ denote the fractional assignment of bundle vertex $P_k$ to agent $i$. 
For $i\in N'$ and $h\in [q]$, let $y_{i,h}$ denote the fractional assignment of high slot $h$ to agent $i$.  
Our goal is to find a fractional perfect matching that satisfies the ex-ante EF constraints, which can then be decomposed into a lottery over integral perfect matchings.
To this end, we introduce the following feasibility LP, which defines the valid polytope of fractional assignments:
\begin{equation*}
\begin{aligned}
\sum_{k \notin N_H} x_{i,k} + \sum_{h=1}^q y_{i,h} &= 1 && \text{for } i \in N', \\
\sum_{i \in N'} x_{i,k} &= 1 && \text{for } k \notin N_H, \\
\sum_{i \in N'} y_{i,h} &= 1 && \text{for } h \in [q], \\
x_{i,k} &= 0 && \text{if } c_i(P_k) > 0, \\
\sum_{h=1}^q y_{i,h} &\leq \sum_{h=1}^q y_{j,h} + \sum_{k \notin N_H} \left ( c_i(P_k) \cdot x_{j,k} \right ) && \text{for } i,j \in N', \\
x_{i,k}, y_{i,h} &\geq 0. &&
\end{aligned}
\end{equation*}

In particular, once a feasible solution $(x,y)$ is decomposed into a lottery over perfect matchings, the second-to-last constraint ensures ex-ante EF among the agents in $N'$.
In the following, we establish the feasibility of this LP.

\begin{lemma}\label{lemma: matching_lp}
    The linear program is feasible and can be solved in polynomial time.
\end{lemma}
\begin{proof}
Let $\mathcal{T}$ be the set of all perfect matchings in the graph, which is non-empty by Lemma~\ref{lemma: matching_configuration}.
In the following, we use the uniform distribution over $\mathcal{T}$ to prove the feasibility of the LP.
Note that the expectation of the uniform distribution induces a feasible solution to the LP, which we denote by $(x,y)$.
The key is to show that the solution satisfies the constraint
$$
\sum_{h=1}^q y_{i,h} \leq \sum_{h=1}^q y_{j,h} + \sum_{k \notin N_H} \left ( c_i(P_k) \cdot x_{j,k} \right ) \quad \text{for } i,j \in N'.
$$
The remaining constraints follow because every matching in $\mathcal{T}$ is a perfect matching in the graph. Fix any two agents $i,j\in N'$ and define the following two collections of perfect matchings:
\begin{itemize}
    \item $\mathcal{A} = \{T \in \mathcal{T} : i \text{ is matched to a high slot}\}$,
    \item $\mathcal{B} = \{T \in \mathcal{T} : j \text{ is matched to a high slot, or to some } P_k \text{ with } c_i(P_k) \geq 1\}$.
\end{itemize}
We construct an injection $\varphi:\mathcal{A} \to \mathcal{B}$ as follows. 
If $T \in \mathcal{A} \cap \mathcal{B}$, we set $\varphi(T) = T$. 
Otherwise, for $T \in \mathcal{A} \setminus \mathcal{B}$, agent $i$ is matched to a high slot while agent $j$ is matched to a bundle vertex $P_k$ with $c_i(P_k) = 0$. 
In this case, we define $\varphi(T)$ as the perfect matching formed by swapping the assignments of $i$ and $j$, so that $i$ is now matched to $P_k$ and $j$ is matched to a high slot. 
Since $c_i(P_k) = c_j(P_k) = 0$ and high slots are adjacent to all agents in $N'$, both newly formed edges belong to the bipartite graph. 
Thus, $\varphi(T)$ is a valid perfect matching belonging to $\mathcal{B}$.
Note that this swap operation is clearly reversible, as swapping the assignments back uniquely recovers $T \in \mathcal{A} \setminus \mathcal{B}$.
Moreover, for any $T \in \mathcal{A} \setminus \mathcal{B}$, its image $\varphi(T)$ matches $i$ to a bundle vertex, whereas for any $T \in \mathcal{A} \cap \mathcal{B}$, $i$ remains matched to a high slot in $\varphi(T)$. Thus, the images of $\mathcal{A} \setminus \mathcal{B}$ and $\mathcal{A} \cap \mathcal{B}$ under $\varphi$ are disjoint, confirming that $\varphi$ is an injection.
Since $\varphi$ is an injection, we have $|\mathcal{A}|\leq|\mathcal{B}|$.
Moreover, every matching has probability $1/|\mathcal{T}|$ under this distribution, and hence
$$
\begin{aligned}
\sum_{h=1}^qy_{i,h}
&=\frac{|\mathcal{A}|}{|\mathcal{T}|}
\leq\frac{|\mathcal{B}|}{|\mathcal{T}|}=\sum_{h=1}^qy_{j,h}
+\sum_{\substack{k\notin N_H\\c_i(P_k)\geq1}}x_{j,k}\leq\sum_{h=1}^qy_{j,h}+\sum_{k\notin N_H} \left ( c_i(P_k) \cdot x_{j,k} \right ),
\end{aligned}
$$
which is exactly the fairness constraint displayed above.
Therefore, the LP is feasible. 
Moreover, since the program involves $O(n^2)$ variables and constraints, and all coefficients are integers bounded by $m$, it can be solved in polynomial time.
\end{proof}

\subsection{Lottery Construction via Birkhoff--von Neumann Decomposition}
Given a feasible solution $(x, y)$ to the LP from Lemma~\ref{lemma: matching_lp}, we next show that it can be decomposed into a convex combination of integral perfect matchings.
To this end, we introduce the Birkhoff--von Neumann decomposition lemma, a standard tool in the BoBW literature for decomposing a fractional assignment matrix into a lottery over integral allocations~\cite{journals/ior/AzizFSV24, conf/atal/0001GM23, journals/jair/HoeferSV24}.
\begin{lemma}[Birkhoff--von Neumann]\label{lemma: birkhoff}
    A matrix is doubly stochastic if it is nonnegative and every row and column sums to one. Every doubly stochastic matrix of dimension $n$ can be decomposed in polynomial time into a convex combination of at most $n^2$ permutation matrices.
\end{lemma}

\begin{lemma}\label{lemma: matching_decomposition}
    Every feasible solution of the linear program can be decomposed in polynomial time into a convex combination of at most $n^2$ perfect matchings in the graph.
\end{lemma}
\begin{proof}
Fix any feasible solution $(x,y)$ to the LP computed by Lemma~\ref{lemma: matching_lp}.
We combine $x$ and $y$ into an $|N'|\times |N'|$ matrix whose rows are indexed by $N'$ and whose columns are indexed by the bundle vertices and high slots, with entries given by the corresponding fractional assignment variables.
The first three groups of constraints imply that this matrix is doubly stochastic. By Lemma~\ref{lemma: birkhoff}, it admits a polynomial-time decomposition
$$
(x,y)=\sum_{\ell=1}^d \left ( \lambda_\ell \cdot T_\ell \right ) \text{ with}
\quad
\sum_{\ell=1}^d\lambda_\ell=1,
$$
where $d\leq n^2$ and every $T_\ell$ is a permutation matrix. 
The positive coefficient $\lambda_\ell$ implies that every entry equal to one in $T_\ell$ corresponds to a positive entry of $(x,y)$. Since $x_{i,k}>0$ only if $c_i(P_k)=0$, every positive entry of $(x,y)$ corresponds to an edge of the graph. 
Therefore, each permutation matrix $T_\ell$ represents a perfect matching in the graph.
\end{proof}

\begin{algorithm}[!htbp]
\caption{Polynomial Additive Feasible Configuration Lottery}
\label{alg: bobw_EFEFX}
Compute a feasible configuration $(\bP,N_H)$ using Lemma~\ref{lemma: feasible_configuration}\;
Let $N_H^+=\{i\in N_H:P_i\neq\emptyset\}$, $q=r-|N_H^+|$, and $N'=N\setminus N_H^+$\;
Construct and solve the LP\;
Decompose $(x,y)$ as $(x,y)=\sum_{\ell=1}^d\lambda_\ell T_\ell$ using Lemma~\ref{lemma: matching_decomposition}\;
Select $T_\ell$ with probability $\lambda_\ell$ and construct the induced configuration $(\tilde \bP, \tilde N_H)$\;
Using fixed orderings of the agents and $M^+$, give $s+1$ chores from $M^+$ to each agent in $\tilde N_H$ and $s$ to every other agent\;
\KwOut{A lottery that is ex-ante EF and fPO, and ex-post EFX and fPO}

\end{algorithm}

\begin{theorem} \label{theorem: Bobw}
    For the allocation of binary additive chores to unweighted agents, Algorithm~\ref{alg: bobw_EFEFX} computes a lottery that achieves ex-ante EF, fPO, and ex-post EFX in polynomial time.
\end{theorem}
\begin{proof}
We first demonstrate that the overall framework runs in polynomial time.
By Lemma~\ref{lemma: feasible_configuration}, a feasible configuration can be efficiently computed.
Based on this configuration, we construct a bipartite graph to characterize the reallocation process.
We then formulate a linear program over all fractional assignments, whose solution can be obtained in polynomial time via Lemma~\ref{lemma: matching_lp}.
Finally, this fractional solution is decomposed into a lottery with a support of at most $n^2$ integral allocations.
Taken together, the total computational complexity remains polynomial.

To complete the proof, we show that every integral allocation in the support is ex-post EFX.
By Lemma~\ref{lemma: matching_configuration}, each such integral allocation induces a valid feasible configuration.
Lemma~\ref{lemma: expost_EFX} then establishes that every realized allocation is EFX and fPO.

It remains to prove that the lottery is ex-ante EF and fPO. 
First, consider any agent $i \in N_H^+$. 
In every allocation in the support, agent $i$ receives bundle $P_i$ and $s+1$ chores from $M^+$, so her realized cost is always $c_i(X_i) = s+1$, giving an expected cost $\mathbb{E}[c_i(X_i)] = s+1$. 
For any other agent $j \in N$, fix a realization $\bX$ induced by a configuration $(\tilde \bP, \tilde N_H)$.
If $j\in \tilde N_H$, $j$ receives $s+1$ chores from $M^+$; if $j \notin \tilde N_H$, $j$ receives $s$ chores from $M^+$ along with $P_j$ where $c_i(P_j) \ge 1$ by feasibility of $(\tilde \bP, \tilde N_H)$.
In both cases, we have $c_i(X_j) \ge s+1$.
Therefore, $\mathbb{E}[c_i(X_j)] \geq \mathbb{E}[c_i(X_i)]$ holds for all $j \in N$.

Next, consider any $i \in N'$ and $j \in N_H^+$. 
By the definition of edges in the bipartite graph, agent $i$ incurs a cost of $s$ when matched to a bundle vertex, and a cost of $s+1$ when matched to a high slot. 
Meanwhile, agent $j$ always receives $s+1$ chores from $M^+$. 
Hence, $c_i(X_i) \leq s+1 \leq c_i(X_j)$ holds in every realization, which implies $\mathbb{E}[c_i(X_i)] \leq \mathbb{E}[c_i(X_j)]$.

Finally, fix $i,j\in N'$. 
Since a bundle vertex $P_k$ can only be assigned to agent $i$ when $c_i(P_k) = 0$, we have
$$
\mathbb{E}[c_i(X_i)] = s + \sum_{h=1}^q y_{i,h}.
$$
On the other hand, the expected cost of agent $i$ for $j$'s bundle is
$$
\mathbb{E}[c_i(X_j)]
=s+\sum_{h=1}^qy_{j,h}+\sum_{k\notin N_H}\left ( c_i(P_k) \cdot x_{j,k} \right ).
$$
The fairness constraint in the LP gives $\mathbb{E}[c_i(X_i)]\leq\mathbb{E}[c_i(X_j)]$.
Therefore, the lottery is ex-ante EF.
For efficiency, the fractional allocation induced by the lottery assigns every chore in $M^0$ only to agents who incur zero cost for it, whereas every chore in $M^+$ contributes one to the total cost. Hence, its social cost is exactly $|M^+|$, which is the minimum possible social cost of any fractional allocation. 
Therefore, the lottery is ex-ante fPO.
\end{proof}

Finally, we show that our technique extends naturally to binary cancelable cost functions.
Although efficiency can no longer be achieved alongside EFX in this broader domain~\cite{journals/tcs/TaoWYZ25}, we can still obtain ex-ante EF and ex-post EFX simultaneously.
\begin{theorem}\label{theorem: cancelable}
For the allocation of binary cancelable chores to unweighted agents, there exists a polynomial-time algorithm that computes a lottery achieving ex-ante EF and ex-post EFX.
\end{theorem}
The corresponding algorithm and its proof are deferred to Appendix~\ref{appendix:cancelable-extension}.

\section{Conclusion and Open Problems}
In this paper, we explore the computation of fair and efficient allocations for indivisible chores in the binary setting. 
For deterministic algorithms, we propose a polynomial-time algorithm to compute allocations that are simultaneously WEFX and fPO. 
For randomized algorithms, we introduce a lottery that ensures ex-ante EF, fPO, and ex-post EFX.
Several open questions deserve future investigation.
First, the existence of EFX allocations for chores under bi-valued instances remains a significant unresolved problem. 
Second, it would be interesting to explore other settings for which ex-ante EF and ex-post EFX can be guaranteed beyond the scope of binary costs.

\section*{Declaration of the Use of AI Tools}
The authors first developed the existence guarantee for both results, and GPT-5.6 Sol further helped to improve the results to polynomial-time computation.
The manuscript was written by the authors, with GPT-5.6 Sol used for language editing and literature searches. 
All arguments were independently verified by the authors, who take full responsibility for the manuscript.

\newpage
\bibliography{ref}
\bibliographystyle{alpha}

\newpage
\appendix
\section{Discussion on Existing Algorithms}\label{appendix:existing-algorithms}
In this appendix, we discuss some natural algorithms from related settings and explain the difficulties of extending them to our setting.

We first consider the algorithm of Wu et al.~\cite{journals/ai/WuZZ25}, which computes WEF1 allocations for additive chores using the \emph{reverse weighted picking sequence} (RWPS).
The algorithm initializes a score $s_i = 0$ for each agent $i$ and constructs a sequence of length $m$ by repeatedly appending the index of an agent with the minimum score\footnote{Note that their algorithm breaks ties by agent index.}, and subsequently increasing her score by $1/w_i$.
The algorithm then reverses this sequence and allows each agent to pick a minimum-cost remaining chore in this reverse order.

We next show that RWPS fails to compute WEFX allocations in our setting.
Consider the following instance.
\begin{example}\label{example: hard}
There are two agents with weights $w_1=2/3$ and $w_2=1/3$, and three chores with the following binary costs.
\begin{center}
\upshape
\begin{tabular}{cccc}
\toprule
\textbf{Agent} & $e_1$ & $e_2$ & $e_3$ \\
\midrule
$1$ & $0$ & $0$ & $1$ \\
$2$ & $1$ & $1$ & $1$ \\
\bottomrule
\end{tabular}
\end{center}
\end{example}

Starting from $(s_1,s_2)=(0,0)$ and breaking the initial tie in favor of agent $1$, RWPS first appends agent $1$ and updates the scores to $(3/2,0)$.
It then appends agent $2$, obtaining $(3/2,3)$, and finally appends agent $1$, obtaining $(3,3)$.
Thus, the forward sequence is $(1,2,1)$, whose reversal is again $(1,2,1)$.
Following this picking order, agent $1$ first takes a zero-cost chore, say $e_1$.
Agent $2$ is indifferent between the two remaining chores; breaking this tie in favor of $e_2$, agent $2$ takes $e_2$ and agent $1$ receives $e_3$.
The resulting allocation is
$$
X_1=\{e_1,e_3\}
\qquad\text{and}\qquad
X_2=\{e_2\},
$$
which is WEF1 but not WEFX.
Hence, RWPS may fail to return a WEFX allocation as it does not control whether an agent's bundle contains a zero-cost chore.
From this example, we observe that we should be very careful in deciding which agents receive items with cost $0$, as such an operation will strengthen the fairness requirement for the agent from WEFX to WEF.

Another natural idea is to borrow the result from the goods setting.
For binary additive goods, Neoh and Teh~\cite{conf/aaai/NeohT25} used a weighted leximin order to establish the existence of WEFX and fPO allocations.
This order adapts the leximin$^{++}$ introduced by Plaut and Roughgarden~\cite{journals/siamdm/PlautR20}. Specifically, it sorts agents by their nondecreasing normalized utilities, defined as $u_i(A_i) = v_i(A_i)/w_i$ where $v_i$ is the valuation function of agent $i$, lexicographically maximizes the resulting vector, and breaks ties using bundle cardinalities.

Therefore, a naive analogue for chores is to lexicographically minimize the normalized cost vector $(c_1(X_1)/w_1, \dots, c_n(X_n)/w_n)$ sorted in nonincreasing order.
However, this analogue already fails for Example~\ref{example: hard}.
Specifically, it uniquely outputs the allocation $\bX = (M, \emptyset)$ with the normalized cost vector $(3/2, 0)$.
This is because assigning any chore to agent 2 results in a normalized cost of at least $1/w_2 = 3$ for her. Although $\bX$ is PO, it is not WEFX.


\section{Counterexamples and Non-existence Proofs for Section~\ref{subsection: extension}}\label{appendix:wefx-nonexistence}
In this appendix, we present the omitted counterexamples and their corresponding analysis for bi-valued and restricted additive instances, as mentioned in Section~\ref{subsection: extension}.

We begin by formally defining both classes of instances. First, we introduce bi-valued instances, a natural extension of binary instances that has been well studied in prior literature.
\begin{definition}[Bi-valued instances]
An additive instance is \emph{bi-valued} if there exist two constants $a>b\geq 0$ such that $c_i(e)\in\{a,b\}$ for every agent $i\in N$ and chore $e\in M$.
\end{definition}

We remark that when $b = 0$, bi-valued instances reduce to binary instances up to scaling\footnote{Notice that all the fairness notions mentioned in this paper are scale-free, i.e., if an allocation satisfies one of these notions, then it remains to satisfy the same notion if we rescale any cost function.}, and a positive value of $b$ moves beyond the binary boundary.

Another generalization of binary instances is restricted additive instances.
\begin{definition}[Restricted additive instances]
An additive instance is \emph{restricted additive} if every chore $e\in M$ has an inherent cost $c(e)\geq 0$ such that $c_i(e)\in\{0,c(e)\}$ for every agent $i\in N$.
\end{definition}

Note that when $c(e)=1$ for every $e \in M$, restricted additive instances also reduce to binary instances.

We are now ready to present our counterexample and its corresponding analysis. Specifically, we first construct a bi-valued instance that fails to admit any WEFX allocation.

\begin{example}[Non-existence of WEFX for bi-valued instances]\label{example:bivalued-nonexistence}
Consider three agents with weights $w_1=1/4$, $w_2=11/24$, and $w_3=7/24$, and six chores $M=A\cup B$, where $A=\{e_1,e_2,e_3,e_4\}$ and $B=\{e_5,e_6\}$.
Their cost functions are given in Table~\ref{tab:bivalued-nonexistence}.
This is a bi-valued instance with values $\{1,4\}$, but it admits no WEFX allocation.
\end{example}

\begin{table}[!htbp]
\centering
\caption{Cost functions in the bi-valued counterexample}
\label{tab:bivalued-nonexistence}
\upshape
\begin{tabular}{ccccccc}
\toprule
& $e_1$ & $e_2$ & $e_3$ & $e_4$ & $e_5$ & $e_6$ \\
\midrule
Agent 1 & $1$ & $1$ & $1$ & $1$ & $4$ & $4$ \\
Agent 2 & $4$ & $4$ & $4$ & $4$ & $1$ & $1$ \\
Agent 3 & $4$ & $4$ & $4$ & $4$ & $1$ & $1$ \\
\bottomrule
\end{tabular}
\end{table}

We now prove the claim in Example~\ref{example:bivalued-nonexistence}.
Fix an allocation $\bX$ and suppose, for contradiction, that it is WEFX.
First, every agent must receive a non-empty bundle. Indeed, suppose towards a contradiction that $X_i = \emptyset$ for some agent $i \in [3]$. 
Since every chore has a strictly positive cost, it must be $|X_j| \leq 1$ for every $j \neq i$ (otherwise, $\hat{c}_j(X_j)/w_j > 0 = c_i(X_i)/w_i$).
Therefore, items can not be completely allocated, leading to a contradiction.
For every agent $i \in [3]$, summing the WEFX inequalities over all $j\neq i$ gives
$$
\begin{aligned}
c_i(M)-c_i(X_i)
=\sum_{j\neq i}c_i(X_j)
&\geq\sum_{j\neq i}\left(\frac{w_j}{w_i}\cdot\hat c_i(X_i)\right)=\frac{1-w_i}{w_i}\cdot\hat c_i(X_i).
\end{aligned}
$$
Equivalently,
$
w_i \cdot c_i(M)\geq w_i \cdot c_i(X_i)+(1-w_i) \cdot \hat c_i(X_i).
$

Let $x_i=|X_i\cap A|$ and $y_i=|X_i\cap B|$.
Substituting the costs and weights into the necessary condition leaves only the following possibilities:
$$
\begin{aligned}
(x_1,y_1)&\in\{(0,1),(1,0),(2,0),(3,0)\},\\
(x_2,y_2)&\in\{(0,1),(0,2),(1,0),(1,1),(1,2),(2,0)\},\\
(x_3,y_3)&\in\{(0,1),(0,2),(1,0),(1,1),(2,0)\}.
\end{aligned}
$$

We claim that agent $1$ cannot receive a chore from $B$.
Suppose for contradiction that agent $1$ receives at least one chore from $B$.
The possibilities above then enforce $(x_1, y_1) = (0, 1)$.
Since $B$ consists of two chores, it follows that $y_2 + y_3 = 1$, which means that exactly one agent $i \in \{2, 3\}$ has $y_i = 1$.
According to the possibilities above, this agent $i$ satisfies $x_i \le 1$.
Meanwhile, the remaining agent $j \in \{2, 3\} \setminus \{i\}$ has $y_j = 0$, which yields $x_j \le 2$.
Combining these bounds with $x_1 = 0$, we get
$$
x_1 + x_2 + x_3 \le 0 + 1 + 2 = 3,
$$
which contradicts the fact that $x_1 + x_2 + x_3 = |A| = 4$.
Thus, we have $y_1=0$ and $y_2+y_3=2$.

Next, we distinguish three cases according to $y_2$.
Each displayed inequality below is the reverse of the corresponding comparison required by WEFX.
\begin{itemize}
    \item Suppose that $y_2=0$ and $y_3=2$.
    The possibilities above give $x_3=0$ and $x_2\in\{1,2\}$.
    If $x_2=1$, then $x_1=3$, $\hat c_1(X_1)=2$, and $c_1(X_2)=1$. 
    Hence, we have
    $$
    \frac{\hat c_1(X_1)}{w_1}
    =8
    >\frac{24}{11}
    =\frac{c_1(X_2)}{w_2}.
    $$
    If $x_2=2$, then $\hat c_2(X_2)=4$ and $c_2(X_3)=2$. Therefore,
    $$
    \frac{\hat c_2(X_2)}{w_2}
    =\frac{96}{11}
    >\frac{48}{7}
    =\frac{c_2(X_3)}{w_3}.
    $$

    \item Suppose that $y_2=y_3=1$.
    We have $x_2,x_3\leq1$ and $x_2+x_3\geq1$.
    If $x_3=1$, then $\hat c_3(X_3)=4$ and $c_3(X_2)\leq5$. Hence,
    $$
    \frac{\hat c_3(X_3)}{w_3}
    =\frac{96}{7}
    >\frac{120}{11}
    \geq\frac{c_3(X_2)}{w_2}.
    $$
    If $x_3=0$, then $x_2=1$, $\hat c_2(X_2)=4$, and $c_2(X_3)=1$. Thus,
    $$
    \frac{\hat c_2(X_2)}{w_2}
    =\frac{96}{11}
    >\frac{24}{7}
    =\frac{c_2(X_3)}{w_3}.
    $$

    \item Suppose that $y_2=2$ and $y_3=0$.
    We have $x_2\leq1$ and $x_3\in\{1,2\}$.
    If $x_3=1$, then $x_1=3-x_2\geq2$, so $\hat c_1(X_1)\geq1$ and $c_1(X_3)=1$. Hence,
    $$
    \frac{\hat c_1(X_1)}{w_1}
    \geq4
    >\frac{24}{7}
    =\frac{c_1(X_3)}{w_3}.
    $$
    If $x_3=2$, then $\hat c_3(X_3)=4$ and $c_3(X_2)\leq6$. Consequently,
    $$
    \frac{\hat c_3(X_3)}{w_3}
    =\frac{96}{7}
    >\frac{144}{11}
    \geq\frac{c_3(X_2)}{w_2}.
    $$
\end{itemize}
Thus, under any allocation, at least one agent violates WEFX, which means that no WEFX allocation exists for this example.
\medskip

Next, we turn our attention to restricted additive instances.
\begin{example}[Non-existence of WEFX for restricted additive instances]\label{example:restricted-nonexistence}
Consider two agents with $w_1=4/9$ and $w_2=5/9$, and $M=\{e_1,e_2,e_3,e_4\}$.
Their cost functions are given in Table~\ref{tab:restricted-nonexistence}.
This instance is restricted additive with $c(e_1)=c(e_2)=1$ and $c(e_3)=c(e_4)=3$, but it admits no WEFX allocation.
\end{example}

\begin{table}[!htbp]
\centering
\caption{Cost functions in the restricted additive counterexample}
\label{tab:restricted-nonexistence}
\upshape
\begin{tabular}{ccccc}
\toprule
& $e_1$ & $e_2$ & $e_3$ & $e_4$ \\
\midrule
Agent 1 & $0$ & $0$ & $3$ & $3$ \\
Agent 2 & $1$ & $1$ & $3$ & $3$ \\
\bottomrule
\end{tabular}
\end{table}

We now prove the claim in Example~\ref{example:restricted-nonexistence}.
Fix an allocation $(X_1, X_2)$ and let $x = |X_1 \cap \{e_1, e_2\}|$.
We analyze three cases based on the number of chores in $\{e_3, e_4\}$ assigned to agent $1$.
As above, each displayed inequality contradicts the corresponding WEFX comparison.
\begin{itemize}
    \item \textbf{Agent $1$ receives neither $e_3$ nor $e_4$.}
    In this case, we have $\hat{c}_2(X_2) \ge 3$ and $c_2(X_1) = x \le 2$.
    From the perspective of agent $2$, we have
    $$
    \frac{\hat{c}_2(X_2)}{w_2} \ge \frac{27}{5} > \frac{9}{2} \ge \frac{c_2(X_1)}{w_1}.
    $$

    \item \textbf{Agent $1$ receives exactly one chore from $\{e_3, e_4\}$.}
    In this case, if $x = 0$, then we have $\hat{c}_2(X_2) = 4$ and $c_2(X_1) = 3$, so
    $$
    \frac{\hat{c}_2(X_2)}{w_2} = \frac{36}{5} > \frac{27}{4} = \frac{c_2(X_1)}{w_1}.
    $$
    If instead $x \ge 1$, then $\hat{c}_1(X_1) = 3$ and $c_1(X_2) = 3$, and hence
    $$
    \frac{\hat{c}_1(X_1)}{w_1} = \frac{27}{4} > \frac{27}{5} = \frac{c_1(X_2)}{w_2}.
    $$

    \item \textbf{Agent $1$ receives both $e_3$ and $e_4$.}
    In this case, we have $\hat{c}_1(X_1) \ge 3$ while $c_1(X_2) = 0$. Therefore,
    $$
    \frac{\hat{c}_1(X_1)}{w_1}\geq\frac{27}{4}>0=\frac{c_1(X_2)}{w_2}.
    $$
\end{itemize}
Hence, no allocation is WEFX.

\section{Cancelable Chores: Ex-ante EF and Ex-post EFX}\label{appendix:cancelable-extension}
In this appendix, we consider cancelable cost functions with binary marginal costs. 
Since Tao et al.~\cite{journals/tcs/TaoWYZ25} showed that EFX and PO are incompatible in this setting, we restrict our focus to ex-ante EF and ex-post EFX.
Throughout this section, for any $S, U \subseteq M$ with $S \cap U = \emptyset$, we define $c_i(U \mid S) = c_i(S \cup U) - c_i(S)$ as the marginal cost of bundle $U$ with respect to $S$. 
For notational simplicity, when $U = \{e\}$ is a singleton, we write $c_i(e \mid S)$ instead of $c_i(\{e\} \mid S)$.

We start with several useful propositions derived from binary cancelable functions.

\begin{proposition}[\cite{journals/tcs/TaoWYZ25}]\label{proposition: cancelable_equality}
    For any $S, Q \subseteq M$ with $c_i(S) = c_i(Q)$, we have $c_i(S \cup U) = c_i(Q \cup U)$ for every set $U \subseteq M$ disjoint from $S \cup Q$.    
\end{proposition}

\begin{proposition}[\cite{journals/tcs/TaoWYZ25}]\label{proposition: submodualrity}
    For any $S \subseteq Q \subseteq M$, if there exists an item $e\notin Q$ with $c_i(e \mid S) = 0$, then $c_i(e \mid Q) = 0$.
    
\end{proposition}

To guarantee ex-post EFX, we similarly begin by computing a balanced partial allocation.
We initialize $A_i = \emptyset$ for each $i \in N$ and $R = M$ as the set of remaining items.
At each step, we update $M^+$ as the set of unallocated chores that incur a marginal cost of $1$ for every agent $i$ with respect to their current bundle $A_i$:
$$
    M^+ = \{e \in R : c_i(e \mid A_i) = 1 \text{ for all } i \in N\}.
$$ 
As long as $|M^+| \ge n$, we select $n$ chores among them and assign exactly one to each agent, and then update $R$.
Let $s = |A_i|$ denote the number of completed rounds (which is identical for all agents $i$).
After computing the balanced partial allocation, we update $M^+$ and define $M^0 = R \setminus M^+$.
Specifically, we let $r = |M^+| < n$.

Now we characterize the instance following this balanced partition.
\begin{lemma}\label{lemma: cancelable_residual}
    The balanced partial allocation can be computed in polynomial time.
    Moreover, $c_i(A_j) = s$ for all $i, j \in N$.
\end{lemma}

\begin{proof}
We prove by induction that, after $t$ completed rounds,
$$
c_i(A_j)=t \quad \text{for all } i,j\in N.
$$
The claim is immediate for $t=0$ as every $A_j$ is empty. 
Suppose that it holds after $t$ rounds, and let $e_j$ be the chore assigned to agent $j$ in the $(t+1)$-th round. Since $e_j\in M^+$, we have $c_i(e_j\mid A_i)=1$ for every agent $i$. The induction hypothesis gives $c_i(A_i)=c_i(A_j)=t$, so Proposition~\ref{proposition: cancelable_equality} implies
$$
c_i(A_j\cup\{e_j\})=c_i(A_i\cup\{e_j\})=t+1.
$$
Thus, the claim also holds after the $(t+1)$-th round. Upon termination after $s$ rounds, we obtain $c_i(A_j)=s$ for all $i,j\in N$.
Since each round removes $n$ chores from the remaining pool and identifying $M^+$ only requires checking all agents over unallocated chores, the phase completes in polynomial time.
\end{proof}

We then refer to the instance resulting from the balanced partition as the \emph{residual instance}.
Below, we extend the concept of feasible configurations to the residual instance.

\begin{definition}[Feasible Configuration in Residual Instance]
    In the residual instance, a configuration $(\bP, N_H)$ is \emph{feasible} if:
    \begin{itemize}
        \item $\bP=(P_1, \dots, P_n)$ is a partition of $M^0$ such that $c_i(P_i \mid A_i) = 0$ for every $i \in N$;
        \item $N_H \subseteq N$ is the set of high agents with $|N_H| = r$;
        \item $c_i(P_j \mid A_i) \geq 1$ for every agent $i \in N_H$ with $P_i \neq \emptyset$ and every agent $j \notin N_H$.
    \end{itemize}
\end{definition}

Similarly, we can efficiently compute such a feasible configuration for the residual instance.

\begin{lemma}\label{lemma: cancelable_configuration}
A feasible configuration for the residual instance can be found in polynomial time.
\end{lemma}
\begin{proof}
    Fix an ordering of the agents, say $1, 2, \ldots, n$.
    We assign each $e\in M^0$ to the first agent $i$ satisfying $c_i(e \mid A_i)=0$, whose existence follows from the definition of $M^0$.
    The resulting partition of $M^0$ is set as $\bP$.
    For each agent $i$, $c_i(P_i \mid A_i) = 0$ naturally follows from Proposition~\ref{proposition: submodualrity}.

    Now we consider the identification of high agents $N_H$.
    If there are at least $r$ agents receiving nothing from $M^0$, then we arbitrarily choose any $r$ agents as $N_H$.
    Otherwise, we set all agents receiving nothing from $M^0$ as high agents.
    We further fill up $N_H$ until $|N_H|=r$, by including the remaining agents following the index order.
    For $i\in N_H$ with $P_i\neq\emptyset$ and $j\notin N_H$, we have $i < j$. 
    Every $e \in P_j$ therefore satisfies $c_i(e \mid A_i)=1$, and monotonicity gives $c_i(P_j \mid A_i) \geq1$.
    The construction uses polynomially many value queries.
\end{proof}

Note that an allocation induced by $(\bP, N_H)$ assigns one distinct chore from $M^+$ to every high agent and none to a low agent, and gives every agent $i \in N$ the resulting residual bundle $P_i$ together with $A_i$.

\begin{lemma}\label{lemma: expost_EFX_cancelable}
Every allocation induced by a feasible configuration is EFX.
\end{lemma}
\begin{proof}
    Fix an induced allocation $\bX$. 
    For any agent $i\notin N_H$ and any $j\in N$, it holds that
    $$
        c_i(X_i) = c_i(P_i \cup A_i) = c_i(P_i \mid A_i) + c_i(A_i) = s \le c_i(X_j),
    $$
    where the last inequality follows from the monotonicity and Lemma~\ref{lemma: cancelable_residual}.
    
    For any agent $i\in N_H$ with $P_i = \emptyset$, we have $c_i(X_i) = c_i(A_i \cup \{e\}) = s+1$, where $e\in M^+$ is the distinct item assigned to the high agent.
    Since there are exactly $s+1$ items in $X_i$, by the binary marginality, we have $\hat c_i(X_i) = s \le c_i(X_j)$ for all $j\in N$.
    It remains to consider an agent $i\in N_H$ with $P_i \neq \emptyset$.
    For such an agent $i$, we have $c_i(X_i) \le c_i(P_i \cup A_i) + 1 = c_i(P_i \mid A_i) + c_i(A_i) + 1 = s+1$.
    For any $j\notin N_H$, it holds 
    $$
        c_i(X_j) = c_i(P_j \cup A_j) = c_i(P_j \cup A_i) = c_i(P_j \mid A_i) + c_i(A_i) \ge s + 1 = c_i(X_i),
    $$
    where the first inequality follows from the construction of $N_H$.
    Finally, for any $j \in N_H$, we have $c_i(X_i) \le s+1 \le c_i(X_j)$.
    Consequently, the induced allocation $\bX$ is EFX.
\end{proof}

We now implement the lottery in polynomial time. 
Fix a feasible configuration $(\bP,N_H)$ returned by Lemma~\ref{lemma: cancelable_configuration}, and let $N_H^+=\{i\in N_H:P_i\neq\emptyset\}$, $q=r-|N_H^+|$, and $N'=N\setminus N_H^+$. 
As in Section~\ref{subsection: matching_reallocation}, construct a bipartite graph whose left side is $N'$ and whose right side contains a bundle vertex $P_k$ for each $k\notin N_H$ and $q$ high slots.
Agent $i\in N'$ is adjacent to $P_k$ if $c_i(P_k \mid A_i)=0$ and to every high slot.

For $i\in N'$, $k\notin N_H$, and $h\in[q]$, use $x_{i,k}$ and $y_{i,h}$ as in Section~\ref{subsection: matching_lp}.
Similarly, consider the following LP:
$$
\begin{aligned}
\sum_{k\notin N_H}x_{i,k}+\sum_{h=1}^qy_{i,h}&=1 &&\text{for }i\in N',\\
\sum_{i\in N'}x_{i,k}&=1 &&\text{for }k\notin N_H,\\
\sum_{i\in N'}y_{i,h}&=1 &&\text{for }h\in[q],\\
x_{i,k}&=0 &&\text{if }c_i(P_k \mid A_i)>0,\\
\sum_{h=1}^qy_{i,h}&\leq\sum_{h=1}^qy_{j,h}+\sum_{k\notin N_H}\left ( c_i(P_k \mid A_i) \cdot x_{j,k} \right ) &&\text{for }i,j\in N',\\
x_{i,k},y_{i,h}&\geq0.&&
\end{aligned}
$$

We next establish the feasibility of the LP and show that it can be solved in polynomial time.
\begin{lemma}\label{lemma: cancelable_lp}
The linear program is feasible, and a feasible solution can be computed in polynomial time.
\end{lemma}
\begin{proof}
Note that the graph always contains a perfect matching obtained by matching each $P_k$ to agent $k$, and the high slots to the agents in $N_H \setminus N_H^+$ arbitrarily.

Let $\mathcal{T}$ be the set of all perfect matchings in the graph. To prove LP feasibility, take the uniform distribution over $\mathcal{T}$. For fixed $i,j\in N'$, let $\mathcal{A}$ contain the matchings in which $i$ receives a high slot, and let $\mathcal{B}$ contain those in which $j$ receives a high slot or some $P_k$ with $c_i(P_k \mid A_i)\geq1$. 

In the following, we construct an injection $\varphi:\mathcal{A}\to\mathcal{B}$. For every $T\in\mathcal{A}\cap\mathcal{B}$, define $\varphi(T)=T$. Now consider any $T\in\mathcal{A}\setminus\mathcal{B}$. Agent $i$ is matched to a high slot, while agent $j$ is matched to some bundle vertex $P_k$ satisfying $c_i(P_k\mid A_i)=0$. 
In this case, we define $\varphi(T)$ as the matching obtained by swapping the assignments of $i$ and $j$. The edge $(i,P_k)$ belongs to the graph by the definition of bundle edges, and agent $j$ is adjacent to every high slot. Thus, $\varphi(T)$ is a valid perfect matching in which $j$ receives a high slot, which implies $\varphi(T)\in\mathcal{B}$.
On $\mathcal{A}\setminus\mathcal{B}$, the swap is reversible and therefore uniquely determines the original matching $T$. Moreover, the image of every matching in $\mathcal{A}\cap\mathcal{B}$ matches $i$ to a high slot, whereas the image of every matching in $\mathcal{A}\setminus\mathcal{B}$ matches $i$ to a bundle vertex. Hence, the images of these two sets are disjoint, proving that $\varphi$ is injective and consequently $|\mathcal{A}|\leq|\mathcal{B}|$.

Let $x$ and $y$ be the average incidence vectors under the uniform distribution over $\mathcal{T}$. Since every matching has probability $1/|\mathcal{T}|$, we have
$$
\begin{aligned}
\sum_{h=1}^qy_{i,h}
&=\frac{|\mathcal{A}|}{|\mathcal{T}|}
\leq\frac{|\mathcal{B}|}{|\mathcal{T}|}=\sum_{h=1}^qy_{j,h}
+\sum_{\substack{k\notin N_H\\c_i(P_k\mid A_i)\geq1}}x_{j,k}\leq\sum_{h=1}^qy_{j,h}
+\sum_{k\notin N_H}\left ( c_i(P_k \mid A_i) \cdot x_{j,k} \right ).              
\end{aligned}
$$
This proves the second-to-last constraint.
All remaining constraints follow directly because $(x, y)$ is a convex combination of perfect matchings.
Since the LP has $O(n^2)$ variables and constraints, and each coefficient $c_i(P_k \mid A_i)$ is an integer bounded between $0$ and $m$, the LP can be solved in polynomial time.
\end{proof}

\begin{lemma}\label{lemma: cancelable_lottery}
Every feasible solution of the linear program can be decomposed in polynomial time into a convex combination of at most $n^2$ perfect matchings in the graph.
\end{lemma}
\begin{proof}
Fix any feasible solution $(x,y)$ of the LP.
Combine $x$ and $y$ into a matrix whose rows are indexed by $N'$ and whose columns are indexed by the bundle vertices and high slots. The first three groups of constraints imply that this matrix is doubly stochastic. By Lemma~\ref{lemma: birkhoff}, it admits a polynomial-time decomposition
$$
(x,y)=\sum_{\ell=1}^d \left(\lambda_\ell\cdot T_\ell\right)
\quad\text{with}\quad
\sum_{\ell=1}^d\lambda_\ell=1,
$$
where $d\leq |N'|^2\leq n^2$, each $\lambda_\ell>0$, and every $T_\ell$ is a permutation matrix. 
Since each $\lambda_\ell>0$, every $T_\ell$ is supported on positive entries of $(x,y)$, all of which correspond to edges of the graph, and therefore every permutation matrix $T_\ell$ represents a perfect matching.
\end{proof}

\begin{algorithm}[!htbp]
\caption{Polynomial Cancelable Feasible Configuration Lottery}
\label{alg: Cancelable_EFX}
Initialize $A_i\gets\emptyset$ for every $i\in N$ and $R\gets M$\;
\tcp{Balance Phase}
Compute $M^+=\{e\in R:c_i(e\mid A_i)=1\text{ for every }i\in N\}$\;
\While{$|M^+|\geq n$}{
Choose any $S\subseteq M^+$ with $|S|=n$, and assign its chores bijectively to the agents\;
Add the assigned chore to each $A_i$, set $R\gets R\setminus S$, and recompute $M^+$\;
}
\tcp{Residual Instance}
Let $s=|A_i|$, $M^0=R\setminus M^+$, and $r=|M^+|$\;
Compute a residual feasible configuration $(\bP,N_H)$ using Lemma~\ref{lemma: cancelable_configuration}\;
Construct and solve the residual LP using Lemma~\ref{lemma: cancelable_lp}, and decompose its solution using Lemma~\ref{lemma: cancelable_lottery}\;
Select a matching according to the decomposition and construct its induced configuration $(\tilde \bP,\tilde N_H)$\;
Set $X_i\gets A_i\cup \tilde P_i$ for every $i$, and add one distinct chore from $M^+$ to $X_i$ for every $i\in \tilde N_H$\;
\KwOut{A lottery that is ex-ante EF and ex-post EFX}
\end{algorithm}

Now we are ready to prove Theorem~\ref{theorem: cancelable}.
By Lemmas~\ref{lemma: cancelable_residual} and~\ref{lemma: cancelable_configuration}, both the balanced phase and a feasible configuration for the residual instance can be computed in polynomial time.
Lemmas~\ref{lemma: cancelable_lp} and~\ref{lemma: cancelable_lottery} then efficiently construct a lottery over perfect matchings.
Therefore, the overall algorithm runs in polynomial time.
Since every perfect matching induces a residual feasible configuration, ex-post EFX follows from Lemma~\ref{lemma: expost_EFX_cancelable}.

It remains to verify ex-ante EF for the allocations induced by these matchings. First, fix $i\in N_H^+$. In every realization, agent $i$ keeps $P_i$ and receives some $e_i\in M^+$, so $X_i=A_i\cup P_i\cup\{e_i\}$. Since $c_i(P_i\mid A_i)=0$ and $c_i(e_i\mid A_i)=1$, Proposition~\ref{proposition: cancelable_equality} gives $c_i(X_i)=s+1$. Consider any agent $j$. If $j$ is high in the realized allocation, then $X_j$ contains $A_j\cup\{e_j\}$ for some $e_j\in M^+$. Lemma~\ref{lemma: cancelable_residual} and Proposition~\ref{proposition: cancelable_equality} give
$$
c_i(A_j\cup\{e_j\})=c_i(A_i\cup\{e_j\})=s+1,
$$
and hence $c_i(X_j)\geq s+1$ by monotonicity. If $j$ is low, then $X_j=A_j\cup P_k$ for some $k\notin N_H$. By feasibility of $(\bP,N_H)$ and Proposition~\ref{proposition: cancelable_equality},
$$
c_i(X_j)=s+c_i(P_k\mid A_j)=s+c_i(P_k\mid A_i)\geq s+1.
$$
Thus, $c_i(X_i)\leq c_i(X_j)$ in every realization.

Next, fix $i\in N'$ and $j\in N_H^+$. If $i$ is matched to a bundle vertex $P_k$, the edge definition gives $c_i(X_i)=s+c_i(P_k\mid A_i)=s$. If $i$ is matched to a high slot, then $X_i=A_i\cup\{e_i\}$ for some $e_i\in M^+$ and $c_i(X_i)=s+1$. Therefore, $c_i(X_i)\leq s+1$. On the other hand, $X_j$ contains $A_j\cup\{e_j\}$ for some $e_j\in M^+$. Lemma~\ref{lemma: cancelable_residual}, Proposition~\ref{proposition: cancelable_equality}, and monotonicity imply
$$
c_i(X_j)\geq c_i(A_j\cup\{e_j\})=c_i(A_i\cup\{e_j\})=s+1.
$$
Hence, $c_i(X_i)\leq c_i(X_j)$ in every realization.

Finally, for $i,j\in N'$, we have
$$
\mathbb E[c_i(X_i)]=s+\sum_{h=1}^qy_{i,h}.
$$
On the other hand, we have
$$
\mathbb E[c_i(X_j)]=s+\sum_{h=1}^qy_{j,h}+\sum_{k\notin N_H}\left ( c_i(P_k \mid A_j) \cdot x_{j,k} \right )
=s+\sum_{h=1}^qy_{j,h}+\sum_{k\notin N_H}\left ( c_i(P_k \mid A_i) \cdot x_{j,k} \right ),
$$
where the last equality follows from Lemma~\ref{lemma: cancelable_residual} and Proposition~\ref{proposition: cancelable_equality}. 
Therefore, the second-to-last constraint of the LP therefore gives $\mathbb E[c_i(X_i)]\leq\mathbb E[c_i(X_j)]$, which completes the proof.



\end{document}